\documentclass[hidelinks,onefignum,onetabnum]{siamart220329}

\makeatletter 
\@mparswitchfalse
\makeatother
\normalmarginpar

\usepackage{lipsum}
\usepackage{amsfonts}
\usepackage{graphicx}
\usepackage{epstopdf}
\usepackage{algorithmic}
\usepackage{bm}
\usepackage{todonotes}
\usepackage{mathtools}
\usepackage{mathrsfs}
\usepackage{xcolor}
\usepackage{chemformula}

\usepackage{tikz}
\usepackage{tikzit}

\tikzstyle{GreenNeuron}=[fill={rgb,255: red,0; green,128; blue,128}, draw={rgb,255: red,0; green,128; blue,128}, shape=circle]
\tikzstyle{BlueNeuron}=[fill={rgb,255: red,0; green,0; blue,173}, draw=none, shape=circle]
\tikzstyle{RedNeuron}=[fill={rgb,255: red,202; green,0; blue,0}, draw=none, shape=circle]
\tikzstyle{BlackCirc}=[fill=black, draw=black, shape=circle, font={\tiny}, scale=0.4]
\tikzstyle{tinyGreen}=[fill={rgb,255: red,0; green,128; blue,128}, draw={rgb,255: red,0; green,128; blue,128}, shape=circle, scale=0.4]
\tikzstyle{tinyRed}=[fill={rgb,255: red,202; green,0; blue,0}, draw={rgb,255: red,202; green,0; blue,0}, shape=circle, scale=0.4]
\tikzstyle{RedDisc}=[fill=none, draw={rgb,255: red,191; green,0; blue,64}, shape=circle, scale=0.75]
\tikzstyle{smallRedDisk}=[fill=none, draw={rgb,255: red,191; green,0; blue,64}, shape=circle, scale=0.6]
\tikzstyle{BigRedDisk}=[fill=none, draw={rgb,255: red,191; green,0; blue,64}, shape=circle, scale=0.9]
\tikzstyle{SuperBigDisckRed}=[fill=none, draw={rgb,255: red,191; green,0; blue,64}, shape=circle, scale=1.05]
\tikzstyle{tinyBlack}=[fill=black, draw=black, shape=circle, scale=0.4]

\tikzstyle{Arrow}=[->, draw=black]
\tikzstyle{Green}=[-, draw={rgb,255: red,0; green,128; blue,128}, thick]
\tikzstyle{GreenDashed}=[-, draw={rgb,255: red,0; green,128; blue,128}, thick, dashed]
\tikzstyle{Dashed}=[-, dashed]
\tikzstyle{RedDashed}=[-, thick, dashed, draw={rgb,255: red,202; green,0; blue,0}]
\tikzstyle{BlueDashed}=[-, draw={rgb,255: red,0; green,0; blue,173}, thick, dashed]
\tikzstyle{dotRed}=[-, draw={rgb,255: red,202; green,0; blue,0}, dotted]
\tikzstyle{Thikblue}=[-, draw={rgb,255: red,0; green,0; blue,173}, thick]
\tikzstyle{ArrowRed}=[draw={rgb,255: red,191; green,0; blue,64}, ->]
\tikzstyle{doubleRedArrow}=[->, draw={rgb,255: red,191; green,0; blue,64}, double]
\tikzstyle{blueArrow}=[draw=blue, ->]
\tikzstyle{spring}=[-, spring]
\tikzstyle{LargeSpring}=[-, largeSpring]
\tikzstyle{thikRedArrow}=[->, thick, draw={rgb,255: red,191; green,0; blue,64}]
\tikzstyle{empty}=[-, draw=none]
\tikzstyle{greenArrow}=[->, draw={rgb,255: red,0; green,128; blue,128}]
\tikzstyle{dubleArrow}=[<->]
\tikzstyle{dashedArrow}=[dashed, ->]
\tikzstyle{double Green Arrow}=[<->, draw={rgb,255: red,0; green,128; blue,128}]

\usetikzlibrary{3d,perspective}
\tikzset
{
  axis/.style={thick,-latex},
  my view/.style={3d view={65}{20}},
  nutation/.style={rotate around x=\nut}
}
\usetikzlibrary{patterns,snakes}
\usetikzlibrary{arrows.meta} 
\tikzstyle{spring}=[line width=0.8,blue!7!black!80,snake=coil,segment amplitude=1,segment length=1.25,line cap=round]
\tikzstyle{largeSpring}=[line width=0.8,blue!7!black!80,snake=coil,segment amplitude=2,segment length=3.25,line cap=round]
\colorlet{my cyan}{cyan!50}
\definecolor{my magenta}{RGB}{0,128,128}
\definecolor{verde}{RGB}{0,128,128}
\usepackage{pgfplots}
\pgfplotsset{compat = newest}
\usetikzlibrary{intersections, pgfplots.fillbetween}
\ifpdf
  \DeclareGraphicsExtensions{.eps,.pdf,.png,.jpg}
\else
  \DeclareGraphicsExtensions{.eps}
\fi

\newsiamremark{remark}{Remark}
\newsiamremark{hypothesis}{Hypothesis}
\crefname{hypothesis}{Hypothesis}{Hypotheses}
\newsiamthm{claim}{Claim}

\headers{Inviscid compressible Leslie--Ericksen equations}{Patrick E.~Farrell,  Giovanni Russo and Umberto Zerbinati}

\title{\mytitle\thanks{Submitted to the editors DATE.
\funding{PEF was supported by the Engineering and Physical Sciences Research Council [EPSRC grants EP/R029423/1 and EP/W026163/1].}}}

\author{
  Patrick E.~Farrell\thanks{Mathematical Institute, University of Oxford, Oxford, United Kingdom 
  (\email{patrick.farrell@maths.ox.ac.uk})} 
  \and
  Giovanni Russo\thanks{Departement of Mathematics, University of Catania, Catania CT, Italy  
  (\email{russo@dmi.unict.it}).} 
  \and
  Umberto Zerbinati\thanks{Mathematical Institute, University of Oxford, Oxford, United Kingdom 
  (\email{umberto.zerbinati@maths.ox.ac.uk})} 
} 

\usepackage{amsopn}

\newcommand{\mytitle}{Kinetic derivation of an inviscid compressible Leslie--Ericksen equation for rarified calamitic gases}
\let \vec \bm
\DeclarePairedDelimiter\abs{\lvert}{\rvert}
\DeclarePairedDelimiter\norm{\lVert}{\rVert}
\newcommand{\cchevrons}[1]{\langle\!\langle #1 \rangle\!\rangle}

\DeclareMathOperator{\tr}{tr}

\newcounter{footnoteInText}
\newcounter{footnoteInNote}
\newcommand{\fnmark}{\stepcounter{footnoteInText}\setcounter{footnote}{\value{footnoteInText}}\addtocounter{footnote}{-1}\footnotemark}
\newcommand{\fntext}[1]{\stepcounter{footnoteInNote}\setcounter{footnote}{\value{footnoteInNote}}\footnotetext{#1}}

\newcommand{\gr}[1]{\textcolor{orange}{#1}}

\colorlet{orange}{black}

\ifpdf
\hypersetup{
  pdftitle={\mytitle}
  pdfauthor={Patrick E.~Farrell and Giovanni Russo and Umberto Zerbinati}
}
\fi

\begin{document}

\maketitle

\begin{abstract}
Nematic ordering describes the phenomenon where anisotropic molecules tend to locally align, like matches in a matchbox. This ordering can arise in solids (as nematic elastomers), liquids (as liquid crystals), and in gases. In the 1940s, Onsager described how nematic ordering can arise in dilute colloidal suspensions from the molecular point of view. However, the kinetic theory of nonspherical molecules has not, thus far, accounted for phenomena relating to the presence of nematic ordering. 

In this work we develop a kinetic theory for the behavior of rarified calamitic (rodlike) gases in the presence of nematic ordering.
Building on previous work by Curtiss, we derive from kinetic theory the rate of work hypothesis that forms the starting point for Leslie--Ericksen theory. We incorporate ideas from the variational theory of nematic liquid crystals to create a moment closure that preserves the coupling between the laws of linear and angular momentum.
The coupling between these laws is a key feature of our theory, in contrast to the kinetic theory proposed by {Curtiss \& Dahler}, where the couple stress tensor is assumed to be zero.
This coupling allows the characterization of anisotropic phenomena arising from the nematic ordering.
Furthermore, the theory leads to an energy functional that is a compressible variant of the classical Oseen--Frank energy (with a pressure-dependent Frank constant) and to an {inviscid} compressible analogue of the Leslie--Ericksen equations. The emergence of compressible aspects in the theory for nematic fluids enhances our understanding of these complex systems.
\end{abstract}



\section{Introduction}
In recent years there has been much interest in the mathematical modeling of liquid crystals~\cite{ball}, both due to their important applications and the mathematical beauty of the arising theories.
Liquid crystals are the core technology of the liquid crystal display industry, and are increasingly employed for the development of novel materials~\cite{majumdarEtAll,caiEtAll,kitstonGeisow}.
In many liquid crystals, nematic ordering arises due to the calamitic (rodlike) nature of the constituent molecules; the short-range interaction potential causes nearby particles to be aligned, or the molecules are aligned by the action of external fields.

In this work we develop a kinetic theory for rarified calamitic gases in the presence of nematic ordering. This is a first step towards a kinetic theory of nematic ordering in calamitic fluids, substances with very low shear modulus whose constituent molecules are rodlike. This class encompasses both many liquid crystals, and many polyatomic gases, such as \ch{N2}, \ch{CO}, \ch{NO}, and \gr{\ch{H2}}.

In particular, we focus on calamitic molecules, as a generic representative of this class. A kinetic theory of nonspherical gases was developed by C.~F.~Curtiss and his collaborators~\cite{curtissI,curtissII,curtissIII,curtissIV}, and was later extended to general nonspherical molecules~\cite{curtissV,dahlerSatherI,dahlerSandlerII}.
An overview of the theory of hard convex body fluids is given by Allen et al.~\cite{allenEtAll}. \text{McCourt et al.}~\cite{McCourt} have applied the Boltzmann--Curtiss and the Waldmann--Snider equations to study the behavior of polyatomic gases under the effect of external fields.

We will build on the original theory proposed by Curtiss, relying on ideas that have been developed for ordered fluids with a nematic nature.
Particularly useful will be ideas developed in the context of the variational theory of nematic liquid crystals~\cite{virga,sonnetVirga,stewart,deGennes}.
These ideas will allow us to develop a moment closure that preserves the coupling between the laws of linear and angular momentum.
The particular closure proposed here will allow the characterization of anisotropic phenomena caused by the emergence of a nematic ordering.

The theory leads to a new, kinetically-motivated energy functional for the nematic configuration that is pressure-dependent. The energy functional is a compressible variant of the familiar one-constant Oseen--Frank energy.
This represents a key difference with the standard variational theories of nematic liquid crystals such as the Oseen--Frank and Landau--de Gennes theories~\cite{oseen,frank,virga,deGennes}. Compressible phenomena in liquid crystals have been studied from the molecular point of view~\cite{gelbartBenShaul,straley}, with a focus on calculating the dependence of the Frank constants on the fluid density.
Our theory differs from the previous work on polyatomic gases of \text{McCourt et al.} in that we study the regime where the nematic ordering is fully established by the action of an external field.

We envision this theory as a first step towards the development of a kinetic theory for dense calamitic fluids in the presence of nematic ordering, which may inform the calculation of material constants required in the standard variational theories.

\section{The Boltzmann--Curtiss equation}

Our starting point will be a reformulation of the Boltzmann equation for calamitic molecules interacting by excluded volume, known as the Boltzmann--Curtiss equation \cite{curtissI,curtissV}:
\begin{equation}
	\partial_t f + \nabla_{\vec{q}}\cdot(\vec{v}f)+\nabla_{\vec{\alpha}}
    \cdot(\dot{\vec{\alpha}}f) = C[f,f],\label{eq:boltzman}
\end{equation}
where the usual configuration space of the Boltzmann equation, consisting of position and velocity $(\vec{q},\vec{v})\in \mathbb{R}^3\times \mathbb{R}^3$,
is enlarged to also include the Euler angles describing the orientation of each molecule and its angular velocity. We employ as configuration space the position and Euler angles, and their time derivatives
$(\vec{q},\vec{v},\vec{\alpha},\dot{\vec{\alpha}})\in \mathbb{R}^{3} \times \mathbb{R}^{3} \times \mathbb{R}^{3} \times \mathbb{R}^{3} \simeq \mathbb{R}^{12}$
(all notation is described in Table \ref{tab:notation})\footnote{We use $(\vec{v}, \dot{\vec{\alpha}})$ instead of $(\vec{p}, \vec{\varsigma})$ as this choice is more convenient for our calculations. The Hessian of the Lagrangian remains positive-definite in these coordinates, ensuring the well-posedness of the Legendre transform. Details are given in the supplementary material, \cref{sec:rationalMechanics}.}.

The collision operator originally proposed for the Boltzmann--Curtiss equation is of the form
\begin{equation} \label{eq:collisionop}
 \color{orange} C[f,g] \coloneqq   \int\!\!\!\!\int\!\!\!\!\int\!\!\!\!\int (g^\prime f^\prime-gf)(\vec{k}\cdot\vec{\mathfrak{g}})dSd\vec{k}d\vec{p}_2d\vec{\alpha}_2d\vec{\varsigma}_2,
\end{equation}
where $\mathfrak{g}$ is the relative velocity of the point of contact and $dS\mathrm{d}\vec{k}$ is the surface element of the excluded volume. The excluded volume of two convex rigid bodies is the volume of the first body that cannot be accessed by the second body due to the presence of the first body, as depicted in Figure \ref{fig:excludedVolume}.
{In the collision operator we choose to integrate with respect to $\vec{\varsigma}$, i.e.~the conjugate moment to the Euler angles, rather than the angular velocity or the total time derivative of the Euler angles. This choice is carefully motivated in Appendix \ref{sec:BoltzmannCurtissCollision}.} 
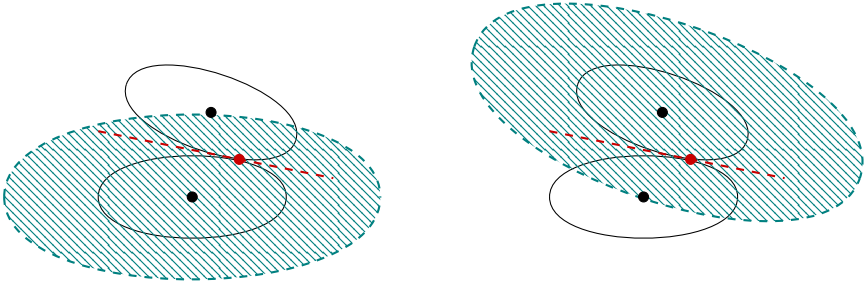
\begin{figure}[h]
    \caption{A section of the excluded volume between two convex molecules is represented by the shaded green region. We also draw the tangent plane to both bodies at the point of contact in red.}
    \centering
    \vspace{-0.5cm}
    \begin{tikzpicture}
	\begin{pgfonlayer}{nodelayer}
		\node [style=none] (0) at (-8.5, -1.5) {};
		\node [style=none] (1) at (-3.5, -1.5) {};
		\node [style=none] (2) at (-7.75, 1.5) {};
		\node [style=none] (3) at (-3.25, 0) {};
		\node [style=none] (4) at (-8.5, 0.25) {};
		\node [style=none] (5) at (-2.25, -1) {};
		\node [style=tinyRed] (6) at (-4.75, -0.5) {};
		\node [style=tinyBlack] (7) at (-5.5, 0.75) {};
		\node [style=tinyBlack] (8) at (-6, -1.5) {};
		\node [style=none] (9) at (-11, -1.5) {};
		\node [style=none] (10) at (-1, -1.5) {};
		\node [style=none] (11) at (3.5, -1.5) {};
		\node [style=none] (12) at (8.5, -1.5) {};
		\node [style=none] (13) at (4.25, 1.5) {};
		\node [style=none] (14) at (8.75, 0) {};
		\node [style=none] (15) at (3.5, 0.25) {};
		\node [style=none] (16) at (9.75, -1) {};
		\node [style=tinyRed] (17) at (7.25, -0.5) {};
		\node [style=tinyBlack] (18) at (6.5, 0.75) {};
		\node [style=tinyBlack] (19) at (6, -1.5) {};
		\node [style=none] (20) at (1.5, 2.5) {};
		\node [style=none] (21) at (11.75, -1) {};
	\end{pgfonlayer}
	\begin{pgfonlayer}{edgelayer}
		\draw [bend left=90, looseness=0.75] (0.center) to (1.center);
		\draw [bend left=270, looseness=0.75] (0.center) to (1.center);
		\draw [bend left=90, looseness=0.75] (2.center) to (3.center);
		\draw [bend left=270, looseness=0.75] (2.center) to (3.center);
		\draw [style=GreenDashed, name path=A, bend right=90, looseness=0.75] (9.center) to (10.center);
		\draw [style=GreenDashed, name path=B, bend left=90, looseness=0.75] (9.center) to (10.center);
		\tikzfillbetween[of=A and B]{pattern=north west lines, pattern color={rgb,255: red,0; green,128; blue,128}};
		\draw [style=RedDashed] (4.center) to (5.center);
		\draw [bend left=90, looseness=0.75] (11.center) to (12.center);
		\draw [bend left=270, looseness=0.75] (11.center) to (12.center);
		\draw [bend left=90, looseness=0.75] (13.center) to (14.center);
		\draw [bend left=270, looseness=0.75] (13.center) to (14.center);
		\draw [style=GreenDashed, name path=C, bend right=90, looseness=0.75] (20.center) to (21.center);
		\draw [style=GreenDashed, name path=D, bend left=90, looseness=0.75] (20.center) to (21.center);
		\tikzfillbetween[of=C and D]{pattern=north west lines, pattern color={rgb,255: red,0; green,128; blue,128}};
		\draw [style=RedDashed] (15.center) to (16.center);
	\end{pgfonlayer}
\end{tikzpicture}
    \label{fig:excludedVolume}
\end{figure}
The pre-collisional and post-collisional distribution functions are denoted by $f, g$ and $f', g'$, respectively.

We know that the collision operator $C[f, f]$ is $L^2$-orthogonal to any invariant preserved in the collision.
In particular, it is known that the collision operator $C[f,f]$ conserves~\cite{stronge}:
\begin{enumerate}
	\item $\psi_1 = 1$, the number of particles in the system;
	\item $\psi_2 = m\vec{v}$, the linear momentum;
	\item $\psi_3 = \mathbb{I}\,\vec{\omega}+\vec{q}\times m\vec{v}$, the angular momentum;
	\item $\psi_4 = \frac{1}{2}m(\vec{v}\cdot \vec{v})+\frac{1}{2}\vec{\omega}\cdot \mathbb{I}\, \vec{\omega}$, the kinetic energy of the system.
\end{enumerate}
\begin{table}[htbp]
	\footnotesize
	\caption{Table of notation.}\label{tab:notation}
  \def\arraystretch{1.2}
		\begin{tabular}{ccc} \hline
			$\vec{q}$ &--& position of the center of mass\\
			$\vec{p}$ &--& linear momentum of the center of mass\\ 
      $\vec{v}$ &--& velocity of the center of mass\\
			$\vec{\omega}$ &--& microscopic angular velocity\\
      $\vec{\alpha}$ &--& Euler angles describing the orientation\\
      $\vec{\varsigma}$ &--& conjugate moment to the Euler angles\\
			$f(\vec{q},\vec{v},\vec{\alpha},\vec{\alpha},t)$ &--
    & solution of the Boltzmann--Curtiss equation\\
      $\vec{\mathfrak{g}}$ & $\vec{p}_1-\vec{p}_2+\vec{\omega}_1\times\vec{g}_1-\vec{\omega}_2\times\vec{g}_2$ & relative velocity of the point of contact\\
			$m$ &--
    & mass of a single molecule\\
			$n(\vec{q},t)$ & $\int\!\!\!\int\!\!\!\int f(\vec{q},\vec{v},\vec{\alpha},\vec{\omega},t) d\vec{v}d\vec{\alpha}d\vec{\omega}$ 
    & particle number density\\
      $\rho(\vec{q},t)$ & $mn(\vec{q},t)$ 
        & fluid mass density\\
			$\cchevrons{\cdot}$ & $\frac{1}{n}\int\!\!\!\int\!\!\!\int\cdot\,f(\vec{q},\vec{v},\vec{\alpha},\vec{\omega},t)d\vec{v}d\vec{\alpha}d\vec{\omega}$ & -- \\ 
			$\vec{v}_0(\vec{q},t)$ & $\cchevrons{\vec{v}}$ 
    & macroscopic stream velocity\\
      $\vec{V}(\vec{q},t)$ & $\vec{v}-\vec{v}_0$ 
        & peculiar velocity\\
      $\vec{\omega}_0(\vec{q},t)$ & $\cchevrons{\vec{\omega}}$ 
        & macroscopic stream angular velocity\\
      $\vec{\Omega}(\vec{q},t)$ & $\vec{\omega} - \vec{\omega}_0$ 
        & peculiar angular velocity\\
      $\vec{\eta}(\vec{q},t)$ & $\cchevrons{\mathbb{I}\vec{\omega}}$ 
        & macroscopic stream angular momentum\\
      $\overline{\vec{\mathbb{I}}}(\vec{q},t)$ & $\cchevrons{\mathbb{I}}$ 
        & macroscopic inertia tensor\\
      \shortstack{$I_1,\,I_2,\,I_3$ \vspace{0.25cm}} & \shortstack{--\vspace{0.25cm}} 
        & \shortstack{nondimensionalised microscopic \\ principal moment of inertia}\\
      $I$ & -- & $3\times 3$ identity matrix\\
      $\mathbb{P}$ & $\cchevrons{\vec{V}\otimes\vec{V}}$ 
        & stress tensor, also known as pressure tensor\\
      $\Pi(\vec{q},t)$ & $\cchevrons{\vec{v}\otimes\vec{v}}$ 
        & linear momentum flux tensor\\
      $\mathbb{M}$ & $\cchevrons{\vec{V}\otimes(\mathbb{I}\vec{\omega})}$
        & couple stress tensor\\
      $\Pi_c(\vec{q},t)$ & $\cchevrons{\vec{v}\otimes (\mathbb{I}\,\vec{\omega})}$ 
        & angular momentum flux tensor\\
      $\vec{\xi}$ & $\cchevrons{\varepsilon_{lki}(nmv_iv_k)\vec{e}_l}$ 
        & vector antisymmetric part of $\cchevrons{\vec{v}\otimes\vec{v}}$\\ 
      \color{orange}$\vec{\mathfrak{Q}}$ & \color{orange}$\frac{1}{2}\cchevrons{\vec{V}\Big(m\abs{\vec{V}}^2+\vec{\Omega}\cdot\mathbb{I}\vec{\Omega}\Big)}$ 
        & \color{orange} heat flux\\
      $\theta$ & $\frac{1}{2}m\abs{\vec{V}}^2+\frac{1}{2}\vec{\Omega}\cdot \mathbb{I}\vec{\Omega}$ 
        & peculiar kinetic energy\\
      $\psi_0(\vec{q},t)$ & $\cchevrons{\theta}$
        & internal energy\\
      $\psi(\vec{q},t)$ & $\frac{1}{2}\Big[m\abs{\vec{v}}^2+\vec{\omega}\cdot\mathbb{I}\vec{\omega}\Big]$ 
        & total energy\\
      $\psi_K(\vec{q},t)$ & $\frac{1}{2} \cchevrons{\Big[m\abs{\vec{v}_0}^2+\vec{\omega}_0\cdot\mathbb{I}\vec{\omega}_0\Big]}$ 
        & macroscopic kinetic energy\\
      $p_{K}(\vec{q},t)$ & $tr[\mathbb{P}^{(0)}]=\frac{\rho}{m}
      I_1 I_2 I_3\cchevrons{\theta}$ 
        & kinetic pressure \\
      $\vec{\nu}$ & \eqref{eq:TensorDecomposition} 
        & nematic director \\
      $k_B$ & $1.380649\times 10^{23}\;J\cdot K^{-1}$ 
        & Boltzmann constant \\
      $N_A$ & $6.02214076 \times 10^{23}\; mol^{-1}$
        & Avogadro constant \\
      $R$ & $N_Ak_B$  & universal gas constant \\
      $\varepsilon_{ijk}$ &--&Levi--Civita symbol\\
      \shortstack{$\mathcal{N}$ \vspace{0.25cm}} & \shortstack{--\vspace{0.25cm}} & \shortstack{number of degrees of freedom in the\\ microscopic Hamiltonian} \\
      \hline
		\end{tabular}
\end{table}
Testing \eqref{eq:boltzman} against a collision invariant $\psi$ and integrating by parts,
since
\begin{equation}
\lim_{\norm{\dot{\vec{\alpha}}}\to \infty}(\psi f)=0,
\end{equation}
we obtain
\begin{equation}
  \partial_t (n\cchevrons{\psi}) + \nabla_{\vec{q}} \cdot (n\cchevrons{\vec{v}\psi})-n\cchevrons{\dot{\vec{\alpha}}\cdot\nabla_{\vec{\alpha}}\psi} = 0,\label{eq:hydrodynamicEquation}
\end{equation}
where the symbol $\cchevrons{\cdot}$ denotes the average over $\vec{p}$, $\vec{\alpha}$ and $\vec{\varsigma}$, i.e.~
\begin{equation}
	\cchevrons{\psi} \coloneqq \frac{1}{n}\int\!\!\!\int\!\!\!\int\psi\,f({\vec{q}},\vec{v},\vec{\alpha},\gr{\vec{\varsigma}},t)d\vec{v}d\vec{\alpha}\gr{d\vec{\varsigma}}.
\end{equation}
Substituting the first collision invariant $\psi_1\equiv 1$ into \eqref{eq:hydrodynamicEquation}, we obtain the equation of conservation of number of particles:
\begin{equation}
  \partial_t n + \nabla_{\vec{q}}\cdot (n\vec{v}_0) = 0.\label{eq:conservationParticleNumber}
\end{equation}
Multiplying by $m$ derives the well-known law of the conservation of mass, i.e.~
\begin{equation}
  \partial_t \rho + \nabla_{\vec{q}}\cdot(\rho\vec{v}_0)=0.\label{eq:conservationMass}
\end{equation}

Substituting the second collision invariant $\psi_2=m\vec{v}$ into \eqref{eq:hydrodynamicEquation}, we obtain
\begin{equation}
  \partial_t(\rho \vec{v}_0)+\nabla_{\vec{q}}\cdot\Big[\rho\cchevrons{\vec{v}\otimes\vec{v}}\Big] = 0,\label{eq:linearMomentumRow}
\end{equation}
where $\cchevrons{\vec{v}\otimes \vec{v}}$ is also known as the linear momentum flux tensor and denoted as $\Pi$.
The above equation is a form of the law of linear momentum well known in continuum mechanics; to rewrite it in a more standard form,
we observe that introducing $\vec{v}_0 = \cchevrons{\vec{v}}$ we have
  \begin{align}
    \Pi = \cchevrons{\vec{v}\otimes\vec{v}}&=\cchevrons{\vec{V}\otimes\vec{V}} + \cchevrons{\vec{v}_0\otimes\vec{V}}+\cchevrons{\vec{V}\otimes\vec{v}_0}+\vec{v}_0\otimes\vec{v}_0\\
    &=\cchevrons{\vec{V}\otimes\vec{V}}+\vec{v}_0\otimes\vec{v}_0,\nonumber
  \end{align}
where $\cchevrons{\vec{v}_0\otimes\vec{V}}=\vec{v}_0 \otimes \cchevrons{\vec{V}}=0$ because $\cchevrons{\vec{V}}=0$.
Now we expand the time derivative in \eqref{eq:linearMomentumRow} and use \eqref{eq:conservationMass} to obtain
\begin{subequations}
  \begin{align}
    (\partial_t \rho)\vec{v}_0 + \rho(\partial_t\vec{v}_0) + \nabla_{\vec{q}}\cdot\Big[\rho \cchevrons{\vec{V}\otimes\vec{V}}+\rho\cchevrons{\vec{v}_0 \otimes \vec{v}_0}\Big]=0,\\ 
    -\Big[\nabla_{\vec{q}}\cdot(\rho \vec{v}_0)\Big]\vec{v}_0 + \rho(\partial_t\vec{v}_0) + \nabla_{\vec{q}}\cdot\Big[\rho \cchevrons{\vec{V}\otimes\vec{V}}+\rho\cchevrons{\vec{v}_0 \otimes \vec{v}_0}\Big]=0,\\
    \rho(\partial_t\vec{v}_0) +\rho\Big[(\nabla_{\vec{q}}\vec{v}_0)\vec{v}_0\Big]+\nabla_{\vec{q}}\cdot\Big[\rho \cchevrons{\vec{V}\otimes\vec{V}}\Big]=0,\\
    \rho \Big[\partial_t \vec{v}_0 + (\nabla_{\vec{q}}\vec{v}_0)\vec{v}_0\Big]+\nabla_{\vec{q}}\cdot (\rho\mathbb{P}) = 0,\label{eq:linearMomentum}
  \end{align}
\end{subequations}
where the last equation is the usual law of linear momentum and $\rho\mathbb{P}$ is the Cauchy stress tensor up to a change of sign.

Substituting the third collision invariant $\psi_3= \mathbb{I} \cdot \vec{\omega}+m(\vec{q}\times \vec{v})$ into \eqref{eq:hydrodynamicEquation}, we obtain
\begin{equation}
  \partial_t (n\vec{\eta}) + \nabla_{\vec{q}}\cdot \Big[n\cchevrons{\vec{v}\otimes \vec{\eta}}\Big]+\partial_t(nm\cchevrons{\vec{q}\times\vec{v}})+\nabla_{\vec{q}}\cdot \Big[nm \cchevrons{\vec{v}\otimes(\vec{q}\times\vec{v})}\Big]=0\label{eq:intermidiateVanishingCrossProduct},
\end{equation}
where $\vec{\eta}=\cchevrons{\mathbb{I}\vec{\omega}}$ is the intrinsic angular momentum field which measures the angular momentum of the fluid per unit mass.
To simplify the last equation we will use the following well-known lemma.
\begin{lemma}[Result 3.5 in \cite{gonzalezStuart}]\label{lem:GonzalezStuart}
  Given a second order tensor $S=S_{ij}\vec{e}_i\otimes\vec{e}_j$, the following identity holds:
  \begin{equation}
    \varepsilon_{lki}\partial_{q_j}(S_{ij})r_k \vec{e}_l = -\varepsilon_{lki} S_{ik}\vec{e}_l+\partial_{q_j}(\varepsilon_{lki}q_kS_{ij})\vec{e}_l.
  \end{equation}
\end{lemma}
We will also use the following lemma. It was used but not proven in \cite{curtissI}, so we include the proof here for completeness.
\begin{lemma}\label{lem:vanishingCrossProduct}
  The following identity holds:
  \begin{equation}
    \partial_t (nm \cchevrons{\vec{q}\times \vec{v}})+\nabla_{\vec{q}} \cdot \Big[ nm \cchevrons{\vec{v}\otimes(\vec{q}\times\vec{v})}\Big]= -\vec{\xi}, 
\end{equation}
where
\begin{equation} \label{eq:definexi}
\vec{\xi} \coloneqq \cchevrons{\varepsilon_{lki}(mnv_iv_k)\vec{e}_l}.
\end{equation}
\end{lemma}
\begin{proof}
We begin by taking a cross product with $\vec{q}$ on both sides of \eqref{eq:linearMomentumRow} to obtain
  \begin{equation}
    \partial_t (\rho\vec{v}_0)\times \vec{q} + \Big[\nabla_{\vec{q}}\cdot(\rho \Pi)\Big]\times \vec{q} = 0.
  \end{equation}
  We then expand $\rho$, $\Pi$ and $\vec{v}_0$ in order to bring the cross product with $\vec{q}$ inside $\cchevrons{\cdot}$ where possible. Furthermore, since $\cchevrons{\cdot}$ does not depend on time we can bring it inside the time derivative to obtain
  \begin{equation}
    \cchevrons{\partial_t(nm\vec{v})\times\vec{q}}+\cchevrons{\nabla_{\vec{q}}\cdot(nm\vec{v}\otimes\vec{v})\times\vec{q}}= 0.
  \end{equation}
where the term $\cchevrons{nm\vec{v}\times\partial_t\vec{q}}$ disappears because $\partial_t\vec{q}\times\vec{v}=0$.
Now we rewrite the divergence term using tensor notation and apply Lemma \ref{lem:GonzalezStuart}:
\begin{equation}
  \partial_{q_j}(nmv_iv_j)q_k\varepsilon_{lki}\vec{e}_l=-\varepsilon_{lki}(nmv_iv_k)\vec{e}_l + \partial_{q_j}(nm\varepsilon_{lki}q_kv_iv_j)\vec{e}_l.
\end{equation}
Defining $\vec{\xi}$ as in \eqref{eq:definexi}, the above expression can be rewritten as:
\begin{equation}
  \cchevrons{\nabla_{\vec{q}}\cdot(nm\vec{v}\otimes\vec{v})\times \vec{q}} = -\vec{\xi}+\nabla_{\vec{q}}\cdot \Big[nm\vec{v}\otimes(\vec{v}\times\vec{q})\Big].
\end{equation}
\end{proof}
Applying now Lemma \ref{lem:vanishingCrossProduct} to \eqref{eq:intermidiateVanishingCrossProduct} we get
\begin{equation}
  \partial_t (n\vec{\eta})+\nabla_{\vec{q}}\cdot \left(n\Pi_c\right) = \vec{\xi},
\end{equation}
where $\Pi_c=\cchevrons{\vec{v}\otimes (\mathbb{I}\,\vec{\omega})}$ is the angular momentum flux tensor and the subscript $\cdot_{c}$ denotes its connection with the couple stress tensor that we will show in a moment.
As we did for the linear momentum flux tensor we would like to rewrite the angular momentum flux tensor in terms of the peculiar velocity, i.e.~
\begin{equation}
  \Pi_c = \cchevrons{\vec{v}_0\otimes (\mathbb{I}\,\vec{\omega})}+\cchevrons{\vec{V}\otimes(\mathbb{I}\,\vec{\omega})}.
\end{equation}
Expanding the time derivative and using \eqref{eq:conservationMass} we obtain
\begin{equation}
  \rho \Big[\partial_t \vec{\eta}+(\nabla_{\vec{q}}\vec{\eta})\vec{v}_0\Big] + \nabla_{\vec{q}}\cdot \left(\rho\mathbb{M}\right) = \vec{\xi},
\end{equation}
where $\mathbb{M}$ is defined as $\cchevrons{\vec{V}\otimes(\mathbb{I}\,\vec{\omega})}$, and it is the couple Cauchy stress tensor up to a change of sign.

In summary, systematically testing \eqref{eq:boltzman} against the first three collision invariants retrieves the usual governing equations of continuum mechanics, i.e.~the conservation laws of mass, of linear momentum, and the balance law of angular momentum:
\begin{subequations}
  \label{eq:continuumMechanics}
  \begin{gather}
    \partial_t \rho + \nabla_{\vec{q}}\cdot (\rho\vec{v}_0) =0,\label{eq:lawMassConservation}\\
    \rho\Big[\partial_t \vec{v}_0 + (\nabla_{\vec{q}}\vec{v}_0)\vec{v}_0\Big] + \nabla_{\vec{q}}\cdot(\rho\,\mathbb{P})=0,\label{eq:lawLinearMomentum}\\
    \rho\Big[\partial_t \vec{\eta} + (\nabla_{\vec{q}}\vec{\eta})\vec{v}_0\Big] +\nabla\cdot(\rho\,\mathbb{M}) = \vec{\xi}.\label{eq:lawAngularMomentum}
  \end{gather} 
\end{subequations}

Now we notice that the law of angular momentum
\eqref{eq:lawAngularMomentum}
can be simplified. Since $\Pi$ is by definition symmetric then $\vec{\xi}=0$, as formalised in the following proposition.
\begin{proposition}
  \label{prop:symmetricPi}
  Since $\Pi$ is symmetric by definition, we know that $\vec{\xi} = 0$.
\end{proposition}
\begin{proof}
  Recall that $\Pi$ is defined as $\cchevrons{\vec{v}\otimes\vec{v}}$, and since $\Pi$ is symmetric we have that
  \begin{equation}
      0=\varepsilon_{lki}nm\Big[\Pi_{ki}-\Pi^T_{ki}\Big]\vec{e}_l=\varepsilon_{lki}nm\Big[\cchevrons{v_iv_k}-\cchevrons{v_kv_i}\Big]\vec{e}_l= 2\vec{\xi},
  \end{equation}
  where in the last equality we have used $\varepsilon_{lki}=-\varepsilon_{lik}$.  
\end{proof}
\begin{corollary}
  Considering equation \eqref{eq:lawAngularMomentum} when $\mathbb{M} = 0$ we have $\rho\dot{\vec{\eta}}=\vec{\xi}$.
  Therefore since $\vec{\xi} = 0$ we have that if $\mathbb{M} = 0$ then the angular momentum is conserved. 
\end{corollary}
Notice that the same observation was also made in \cite{curtissV}, after which \text{Curtiss} and \text{Dahler} wrote:
\begin{quote}
From this, it follows that in the absence of external torques the spin angular
momentum density of a dilute gas is of no physical
consequence since it is coupled neither to the surroundings of the system nor to the density, velocity, and temperature fields of the fluid. Because of this we may
choose the $\mathbb{M}$ (and hence $\vec{\omega_0}$) to be identically zero.
\end{quote}
The next sections of the chapter will be devoted to developing a theory that will allow us to refute this conclusion.

\section{Leslie--Ericksen rate of work hypothesis}
\label{sec:LeslileRateOfWork}
To take into account the anisotropic behaviour that stems from the equation of angular momentum, our objective will be to find a coupling between \eqref{eq:lawLinearMomentum} and \eqref{eq:lawAngularMomentum}. Indeed, in the absence of such coupling, one can assume $\mathbb{M}$ is identically zero.
If $\mathbb{M}$ vanishes this will result in the inability of rarefied calamitic fluids to undergo reorientation under ultrasonic wave propagation, a phenomenon that has been observed in other denser nematic fluids \cite{bertolottiEtAll,scudieriEtAll}.
It is worth mentioning that such coupling exists in the standard dynamic theory of liquid crystals, arising as the Ericksen tensor in the Leslie--Ericksen equations.

For this reason, our first task will be to connect \eqref{eq:continuumMechanics} with the Leslie--Ericksen theory.
In order to achieve this we consider the fourth collision invariant $\psi_4=\frac{1}{2}m (\vec{v}\cdot\vec{v})+\frac{1}{2}\vec{\omega}\cdot \mathbb{I}\vec{\omega}$ and substitute it into \eqref{eq:hydrodynamicEquation}:
\begin{equation}
    \partial_t \Big(n\cchevrons{\frac{1}{2}m(\vec{v}\cdot\vec{v})+\frac{1}{2}\vec{\omega}\cdot\mathbb{I}\vec{\omega}}\Big)+\nabla_{\vec{q}}\cdot\Big[n\cchevrons{\vec{v}\frac{1}{2}m(\vec{v}\cdot\vec{v})+\vec{v}\frac{1}{2}\vec{\omega}\cdot\mathbb{I}\vec{\omega}}\Big]=0.
\end{equation}
We now rewrite the previous expression using the notion of peculiar velocity. In particular, using the identity $\cchevrons{\vec{v}\cdot\vec{v}} = \cchevrons{\vec{V}\cdot\vec{V}}+\cchevrons{\vec{v}_0\cdot \vec{v}_0}$ yields
\begin{align}
\partial_t \Big(n\cchevrons{\frac{1}{2}m (\vec{V}\cdot\vec{V})}\Big)&+\partial_t \Big(n\cchevrons{\frac{1}{2}\vec{\omega}\cdot \mathbb{I}\vec{\omega}}\Big)+\partial_t \Big(n\cchevrons{\frac{1}{2}m(\vec{v}_0\cdot \vec{v}_0)}\Big)\label{eq:energyHydrodynamicEq}\\
    &+\nabla_{\vec{q}}\cdot\Big[n\cchevrons{\vec{v}_0 \frac{1}{2}m(\vec{v}_0\cdot\vec{v}_0)}\Big]+\nabla_{\vec{q}}\cdot\Big[n\vec{v}_0\cchevrons{\frac{1}{2}m(\vec{V}\cdot\vec{V})}\Big]\nonumber\\
    &+\nabla_{\vec{q}}\cdot\Big[n\cchevrons{\vec{V}\frac{1}{2}m(\vec{V}\cdot\vec{V})}\Big]+\nabla_{\vec{q}}\cdot\Big[n\cchevrons{\vec{v}\frac{1}{2}\vec{\omega}\cdot\mathbb{I}\vec{\omega}}\Big]\nonumber\\
    &+\nabla_{\vec{q}}\cdot\Big[nm\mathbb{P}\vec{v}_0\Big]+\nabla_{\vec{q}}\cdot\Big[n\cchevrons{\vec{V}\frac{1}{2}(\vec{v}_0\cdot\vec{v}_0)}\Big]=0\nonumber.
\end{align}
First we notice that the last term vanishes because $\cchevrons{\vec{V}(\vec{v}_0\cdot\vec{v}_0)}=\cchevrons{\vec{V}}(\vec{v}_0\cdot\vec{v}_0)=0$. Next we develop a portion of the terms involved in the previous equation, i.e.~
\begin{align}
  &\partial_t \Big(n\frac{1}{2}m(\vec{v}_0\cdot\vec{v}_0)\Big)+\nabla_{\vec{q}}\cdot\Big[n\vec{v}_0\frac{1}{2}m(\vec{v}_0\cdot\vec{v}_0)\Big]\\
  &= nm(\partial_t\vec{v}_0)\cdot\vec{v}_0 + m(\partial_t n)\frac{\vec{v}_0\cdot\vec{v}_0}{2}+\nabla_{\vec{q}}\cdot\Big[n{\vec{v}_0\frac{1}{2}m(\vec{v}_0\cdot\vec{v}_0)}\Big]\nonumber\\
  &= nm(\partial_t\vec{v}_0)\cdot\vec{v}_0 -\nabla_{\vec{q}}\cdot\Big[nm\vec{v}_0\Big]\frac{\vec{v}_0\cdot\vec{v}_0}{2}+\nabla_{\vec{q}}\cdot\Big[n{\vec{v}_0\frac{1}{2}m(\vec{v}_0\cdot\vec{v}_0)}\Big]\nonumber\\
  &= nm(\partial_t\vec{v}_0)\cdot\vec{v}_0 + nm\Big[(\nabla_{\vec{q}}\vec{v}_0)\vec{v}_0\Big]\cdot\vec{v}_0 =-\vec{v}_0\cdot \Big(\nabla_{\vec{q}}\cdot(nm\mathbb{P})\Big),\nonumber 
\end{align}
where we have employed the chain rule, Leibniz rule, \eqref{eq:lawMassConservation}, and the scalar product of \eqref{eq:lawLinearMomentum} with $\vec{v}_0$.
Substituting the previous equation inside \eqref{eq:energyHydrodynamicEq} we obtain the following expression:
\begin{align}
  \partial_t \Big(n\cchevrons{&\frac{1}{2}m\vec{V}\cdot \vec{V}}\Big)\!+\partial_t\Big(n\cchevrons{\frac{1}{2}\vec{\omega}\cdot\mathbb{I}\vec{\omega}}\Big)-\vec{v}_0\!\cdot\!\Big(\nabla_{\vec{q}}\cdot(nm\mathbb{P})\Big)+\nabla_{\vec{q}}\cdot\Big[n\cchevrons{\frac{\vec{v}_0}{2}m(\vec{V}\cdot\vec{V})}\Big]\label{eq:energyHydrodynamicEqUnexpandedPressureTensor}\\
  &+\nabla_{\vec{q}}\cdot\Big[n\cchevrons{\vec{V}\frac{1}{2}m(\vec{V}\cdot\vec{V})}\Big]+\nabla_{\vec{q}}\cdot\Big[n\cchevrons{\vec{v}\frac{1}{2}\vec{\omega}\cdot\mathbb{I}\vec{\omega}}\Big]+\nabla_{\vec{q}}\cdot\Big[nm\mathbb{P}\vec{v}_0\Big]=0.\nonumber
\end{align}
Combining \eqref{eq:energyHydrodynamicEqUnexpandedPressureTensor} with the following well-known identity \cite[Result 2.11]{gonzalezStuart}
\begin{equation}
\nabla_{\vec{q}}\cdot (S^T\vec{v}_0)=(\nabla_{\vec{q}}\cdot S)\cdot\vec{v}_0+S:\nabla_{\vec{q}}\vec{v}_0,\label{eq:tensorDivergence}
\end{equation}
we obtain the following equation:
\begin{align}
  \partial_t \Big(n\cchevrons{\frac{1}{2}m\vec{V}\cdot \vec{V}}\Big)&+\partial_t\Big(n\cchevrons{\frac{1}{2}\vec{\omega}\cdot\mathbb{I}\vec{\omega}}\Big)+\Big(nm\mathbb{P}\Big):\nabla_{\vec{q}} \vec{v}_0 +\nabla_{\vec{q}}\cdot\Big[n\cchevrons{\frac{\vec{v}_0}{2}m(\vec{V}\cdot\vec{V})}\Big]\\
  &+\nabla_{\vec{q}}\cdot\Big[n\cchevrons{\vec{V}\frac{1}{2}m(\vec{V}\cdot\vec{V})}\Big]+\nabla_{\vec{q}}\cdot\Big[n\cchevrons{\vec{v}\frac{1}{2}\vec{\omega}\cdot\mathbb{I}\vec{\omega}}\Big]=0.\nonumber
\end{align}
We are now left to focus our attention on the terms involving the angular velocity. In particular we begin rewriting the previous equations in terms of the peculiar angular velocity, using the fact that $\cchevrons{\mathbb{I}\vec{\Omega}}$ vanishes:
\begin{align}
\partial_t &\Big(n\cchevrons{\frac{1}{2}m\vec{V}\cdot\vec{V}}\Big)+\partial_t\Big(n\cchevrons{\frac{1}{2}\vec{\omega}_0\cdot\mathbb{I}\vec{\omega}_0}\Big)+\partial_t\Big(n\cchevrons{\frac{1}{2}\vec{\Omega}\cdot\mathbb{I}\vec{\Omega}}\Big)+\Big(nm\mathbb{P}\Big):\nabla \vec{v}_0\label{eq:energyHydrodynamicEqUnexpandedPressureTensorWithPeculiarAngularVelocity}\\
  &+\nabla_{\vec{q}}\cdot\Big[n\cchevrons{\frac{\vec{v}_0}{2}m\abs{\vec{V}}^2}\Big]+\nabla_{\vec{q}}\cdot\Big[n\cchevrons{\frac{\vec{V}}{2}m\abs{\vec{V}}^2}\Big]\nonumber\\
  &+\nabla_{\vec{q}}\cdot\Big[n\cchevrons{\vec{v}_0\frac{1}{2}\vec{\omega}_0\cdot\mathbb{I}\vec{\omega}_0}\Big]+\nabla_{\vec{q}}\cdot\Big[n\cchevrons{\vec{v}_0\frac{1}{2}\vec{\Omega}\cdot\mathbb{I}\vec{\Omega}}\Big]\nonumber\\
  &+\nabla_{\vec{q}}\cdot\Big[n\cchevrons{\vec{V}\frac{1}{2}\vec{\Omega}\cdot\mathbb{I}\vec{\Omega}}\Big]+\nabla_{\vec{q}}\cdot\Big[n\cchevrons{\vec{V}\frac{1}{2}\vec{\omega}_0\cdot\mathbb{I}\vec{\omega}_0}\Big]+\nabla_{\vec{q}}\cdot\Big(n\mathbb{M}^T\vec{\omega}_0\Big)=0,\nonumber
\end{align}
where the second term in the last row vanishes as before because $\cchevrons{\vec{V}} = 0$.
Now we develop a portion of the terms involved in the previous equation, i.e.~
\begin{align}
  &\partial_t \Big(n \cchevrons{\frac{1}{2}\vec{\omega}_0\cdot\mathbb{I}\vec{\omega}_0}\Big)+\nabla_{\vec{q}}\cdot \Big[n\vec{v}_0\cchevrons{\frac{1}{2}\vec{\omega}_0\cdot\mathbb{I}\vec{\omega}_0}\Big]\label{eq:energyLawAngularMomentum}\\
  &=(\partial_t n)\frac{1}{2}\cchevrons{\vec{\omega}_0\cdot \mathbb{I}\vec{\omega}_0}+n\partial_t\Big(\cchevrons{\mathbb{I}\vec{\omega}_0}\Big)\cdot\vec{\omega}_0+\nabla_{\vec{q}}\cdot \Big[n\vec{v}_0\frac{1}{2}\vec{\omega}_0\cdot\cchevrons{\mathbb{I}\vec{\omega}_0}\Big]\nonumber\\
  &=-\frac{1}{2}\nabla_r\cdot(n\vec{v}_0)\cchevrons{\vec{\omega}_0\cdot \mathbb{I}\vec{\omega}_0}+n\partial_t\Big(\cchevrons{\mathbb{I}\vec{\omega}_0}\Big)\cdot\vec{\omega}_0+\nabla_{\vec{q}}\cdot \Big[n\vec{v}_0\frac{1}{2}\vec{\omega}_0\cdot\cchevrons{\mathbb{I}\vec{\omega}_0}\Big]\nonumber\\
  &= n\partial_t\Big(\cchevrons{\mathbb{I}\vec{\omega}_0}\Big)\cdot \vec{\omega}_0 + n\Big[\nabla_{\vec{q}}\big(\cchevrons{\mathbb{I}\vec{\omega}_0}\big)\vec{v}_0\Big]\cdot \vec{\omega}_0=-\vec{\omega}_0 \cdot\Big(\nabla_{\vec{q}}\cdot (n\mathbb{M})\Big)\nonumber\\
  &= \left(n\mathbb{M}\right):\nabla_{\vec{q}}\vec{\omega}_0 - \nabla_{\vec{q}}\cdot\Big(n\mathbb{M}^T\vec{\omega}_0\Big)\nonumber,
\end{align}
where we employ the chain rule, the Leibniz rule, \eqref{eq:conservationParticleNumber}, the scalar product of \eqref{eq:lawAngularMomentum} with $\vec{\omega}_0$, and \eqref{eq:tensorDivergence}.
We now substitute the previous equation into \eqref{eq:energyHydrodynamicEqUnexpandedPressureTensorWithPeculiarAngularVelocity} and obtain 
\begin{align}
\partial_t \Big(n\cchevrons{\frac{1}{2}m\vec{V}\cdot \vec{V}}\Big)&+\partial_t\Big(n\cchevrons{\frac{1}{2}\vec{\Omega}\cdot\mathbb{I}\vec{\Omega}}\Big)+\Big(nm\mathbb{P}\Big):\nabla \vec{v}_0\\
  &+\left(n\mathbb{M}\right):\nabla_{\vec{q}}\vec{\omega}_0\nonumber\\\
  &+\nabla_{\vec{q}}\cdot\Big[n\cchevrons{\frac{\vec{v}_0}{2}m\abs{\vec{V}}^2}\Big]+\nabla_{\vec{q}}\cdot\Big[n\cchevrons{\frac{\vec{V}}{2}m\abs{\vec{V}}^2}\Big]\nonumber\\
  &+\nabla_{\vec{q}}\cdot\Big[n\cchevrons{\vec{v}_0\frac{1}{2}\vec{\Omega}\cdot\mathbb{I}\vec{\Omega}}\Big]+\nabla_{\vec{q}}\cdot\Big[n\cchevrons{\vec{V}\frac{1}{2}\vec{\Omega}\cdot\mathbb{I}\vec{\Omega}}\Big]=0.\nonumber
\end{align}
Finally, we rewrite the previous equation using the heat flux and the kinetic energy $\theta$, respectively
\begin{equation}
  \vec{\mathfrak{Q}} = \frac{1}{2}\cchevrons{\vec{V}\Big(m\abs{\vec{V}}^2+\vec{\Omega}\cdot\mathbb{I}\vec{\Omega}\Big)}, \qquad \theta = \frac{1}{2}m\abs{\vec{V}}^2+\frac{1}{2}\vec{\Omega}\cdot \mathbb{I}\vec{\Omega}, 
\end{equation}
to obtain the following theorem of power expended:
\begin{align}
  \partial_t \Big(n \cchevrons{\theta}\Big)+\nabla_{\vec{q}}\cdot \Big[n\vec{v}_0\cchevrons{\theta}\Big]+(\rho\mathbb{P}):\nabla_{\vec{q}}\vec{v}_0&+(n\mathbb{M}):\nabla_{\vec{q}}\vec{\omega}_0+\nabla_{\vec{q}}\cdot \vec{\mathfrak{Q}}\label{eq:equationOfPowerExpandedWithEnergyFlux}= 0.
\end{align}

Using \eqref{eq:conservationParticleNumber} together with the Leibniz rule we can rewrite \eqref{eq:equationOfPowerExpandedWithEnergyFlux} as:
\begin{align}
  n\Big(\partial_t\cchevrons{\theta}\Big)+\nabla_{\vec{q}}\cdot\Big[n\vec{v}_0\cchevrons{\theta}\Big]-\nabla_{\vec{q}}\cdot \Big[n\vec{v}_0\Big]\cchevrons{\theta}&
  +(\rho\mathbb{P}):\nabla_{\vec{q}}\vec{v}_0 \\
  \nonumber &+(n\mathbb{M}):\nabla_{\vec{q}}\vec{\omega}_0 + \nabla_{\vec{q}}\cdot \vec{\mathfrak{Q}}= 0,
\end{align}
which can be rewritten by making use of the definition of the total derivative of an Eulerian field and the Leibniz rule as:
\begin{equation}
  \rho\dot{\psi_0}+(\rho m\mathbb{P}):\nabla_{\vec{q}}\vec{v}_0+(\rho\mathbb{M}):\nabla_{\vec{q}}\vec{\omega}_0+\nabla_{\vec{q}}\cdot\vec{\mathfrak{Q}}= 0,\label{eq:equationOfPowerExpanded}
\end{equation}
where $\psi_0(\vec{q},t)$ is the internal energy and is defined as $\psi_0(\vec{q},t)=\cchevrons{\theta}$.

A key observation of this work is the fact that \eqref{eq:equationOfPowerExpanded}, when integrated over a domain $\Omega\subset\mathbb{R}^3$, is similar to the rate of work hypothesis that forms the starting point of the modern development of Leslie--Ericksen theory for the dynamics of liquid crystals.
We will show in the next sections that \eqref{eq:equationOfPowerExpanded} relates to the rate of work hypothesis postulated in \cite{ericksen2} and is precisely the same rate of work hypothesis presented in \cite{leslie2}.
Last we would like to point out that \eqref{eq:equationOfPowerExpanded} is not only a reconciliation between the well-established continuum mechanics theory of nematic fluids and the kinetic theory approach here proposed but also a validation of the rate of work hypothesis postulated in \cite{leslie2}.
In fact rather than using the integral version of \eqref{eq:equationOfPowerExpanded} as an assumption, as in \cite{ericksen2} and \cite{leslie2}, the rate of work hypothesis has been derived starting from the foundations of Boltzmann--Curtiss kinetic theory.
\section{Maxwellian distribution}
In this section we derive constitutive relations for the stress tensor, the couple stress tensor, and the internal energy from the Maxwellian distribution for the Boltzmann--Curtiss equation~\eqref{eq:boltzman}, i.e.~that $f^{(0)}$ satisfying $C[f^{(0)},f^{0}]=0$. In seminal work by Curtiss~\cite{curtissI,curtissV}, this distribution was obtained using formal computations from statistical mechanics.
The Maxwellian is given by
\begin{equation}
  f^{(0)}(\vec{\alpha},\vec{V},\vec{\Omega})=\frac{nQ\sin(\alpha_2)}{\int_{\mathbb{R}^3}Q\sin(\alpha_2)\,d\alpha}\frac{m^{\frac{3}{2}}(I_1I_2I_3)^{\frac{1}{2}}}{(\frac{4}{\mathcal{N}}\pi\cchevrons{\theta})^3}\exp\Big[-m\frac{\abs{\vec{V}}^2}{\frac{4}{\mathcal{N}}\cchevrons{\theta}}-\frac{\vec{\Omega}\cdot\mathbb{I}\vec{\Omega}}{\frac{4}{\mathcal{N}}\cchevrons{\theta}}\Big]\label{eq:Maxwellian}
\end{equation}
where $Q=\exp(\frac{\vec{\omega}_0\cdot \mathbb{I}\vec{\omega}_0}{\frac{2}{3}\cchevrons{\theta}})$ and both $\omega_0$ and $\theta$ are assumed to remain constant at equilibrium.
Notice that the distribution here presented is the ``absolute Maxwellian'' since it is spatially homogeneous, i.e.~we are assuming that at equilibrium $n,\vec{V},\vec{\Omega},\theta$ are spatially independent quantities.

The only difference between the formulation of the Maxwellian considered here and the one presented in \cite{curtissI,curtissV} is that \eqref{eq:Maxwellian} is expressed in terms of internal energy $\cchevrons{\theta}$ rather than in terms of kinetic temperature.
If one is interested in obtaining \eqref{eq:Maxwellian} from \cite{curtissI,curtissV} one only has to apply the equipartition of energy theorem, i.e.~
\begin{equation}
  \cchevrons{\theta} = \frac{\mathcal{N}}{2} k_BT,\label{eq:equipartitionOfEnergy}
\end{equation}
where $T$ is the kinetic temperature, $\mathcal{N}$ is the number of degrees of freedom appearing in the Hamiltonian of a single molecule and $k_B$ is the Boltzmann constant.

What is the correct value of $\mathcal{N}$ for a polyatomic gas? To answer this question we consider linear diatomic gases, such as \ch{H2}, \ch{N2} and \ch{O2}.
We can think of these polyatomic gases as two spheres connected by a spring and therefore the Hamiltonian of this model would have seven degrees of freedom: the moments of the centers of mass of the two spheres, and the length of the spring connecting the two centers of mass.
Yet we know that vibration modes for the previously mentioned diatomic molecules are excited only above $2000K$, a temperature well above the regime of interest here.
This suggests that we can ignore the degrees of freedom corresponding to the spring connecting the center of mass of each atom and model the polyatomic gas molecule as a rigid body.
Lastly, in a later section, the hypothesis that our molecules are slender bodies will be essential to ignore some terms in the expansion of the inertia tensor.
In the slender body limit, we can regard our molecules as segments in space, thus losing one rotational degree of freedom in the kinetic energy; more details are provided in the supplementary material, \ref{sec:rationalMechanics}.
This suggests that in our regime of interest, we have $\mathcal{N}=5$ and the correct relation arising from the equipartition of energy theorem would be
\begin{equation}
  \label{eq:polyEquipartitionOfEnergy}
  \cchevrons{\theta} = \frac{5}{2} k_BT.
\end{equation}

We notice that the emergence of a nematic ordering implies that the remaining five rotational degrees of freedom are not independent among molecules and therefore the heat capacity of the nematic gas should be less than $\frac{5}{2}R$ per mole, where $R$ denotes the gas constant.
These difficulties in understanding the correct way to apply the equipartition of energy theorem are reflected in the experimental observation of the heat capacity of nematic fluids such as liquid crystals \cite{heatCapacity}.
In particular, in \cite{heatCapacity} the author observes a sharp increase in the specific heat capacity of the liquid crystal MBBA at constant pressure when the temperature is increased from $319K$ to $324K$.
Such behaviour is anomalous compared to what is observed for polyatomic gases such as \ch{N2}, for which the specific heat capacity at constant pressure is constant between $123K$ and $573K$ \cite{chase}.
For this reason, we decided to treat non-spherical rarefied gases in the presence of nematic alignment in a standard manner, i.e.~assuming \eqref{eq:polyEquipartitionOfEnergy}, and devote future work to this particular issue.

We now proceed to use the Maxwellian distribution \eqref{eq:Maxwellian} to obtain constitutive relations for the stress tensor, the couple stress tensor, and the internal energy.
We can compute the pressure tensor near equilibrium by evaluating $\cchevrons{\vec{V}\otimes\vec{V}}$ using the Maxwellian $f^{(0)}$ as a probability distribution. 
It turns out that the pressure tensor near equilibrium is nothing more than the variance of the Maxwellian distribution, i.e.~
\begin{equation}
  \mathbb{P}^{(0)} = \cchevrons{\vec{V}\otimes\vec{V}}_{\sim f^{(0)}}=\frac{2(I_1I_2I_3)^{\frac{1}{2}}}{5m}\cchevrons{\theta}I.
\end{equation}

Following the same reasoning we can compute also the couple stress tensor, which turns out to be zero:
\begin{equation}
  \mathbb{M}^{(0)} = \cchevrons{\vec{V}(\mathbb{I}\vec{\omega})}_{\sim f^{(0)}}=\mathbb{E}_{f^{(0)}_{\vec{\Omega}}}\big[\mathbb{I}\vec{\omega}\big]\mathbb{E}_{f^{(0)}_{\vec{V}}}\big[\vec{V}\big]= 0.
\end{equation}
This follows from the Fubini--Tonelli Theorem, the fact that $f^{(0)}=f^{(0)}_{\vec{V}}f^{(0)}_{\vec{\Omega}}f^{(0)}_{\vec{\alpha}}$, and that $f^{(0)}_{\vec{V}}$ has zero mean.
In particular, since we know $\vec{\xi}$ is also zero, we can use \eqref{eq:lawAngularMomentum} to deduce that $\dot{\vec{\eta}}=0$.
In what follows, we will be interested in weighting the pressure tensor and the couple stress tensor by the fluid density, and for this reason, we introduce the kinetic pressure 
\begin{equation}
  p_K\coloneqq tr[\rho\mathbb{P}^{(0)}]=\frac{6}{5}\frac{\rho}{m}(I_1I_2I_3)^{\frac{1}{2}}\cchevrons{\theta}.\label{eq:kineticPressure}
\end{equation}

\section{Oseen--Frank energy}
\label{sec:OseenFrankEnergy}
Since we now know that energy stored in the particle ensemble is given by $\psi=\cchevrons{\frac{1}{2}m\abs{\vec{v}}^2+\frac{1}{2}\vec{\omega}\cdot\mathbb{I}\vec{\omega}}$ we can proceed by rewriting,
\begin{align}
  \psi = \frac{1}{2}\cchevrons{m\abs{\vec{v}}^2+\vec{\omega}\cdot\mathbb{I}\vec{\omega}} &= \frac{m}{2}\cchevrons{\abs{\vec{V}}^2}+\frac{m}{2}\abs{\vec{v}_0}^2+\frac{1}{2}\cchevrons{\vec{\Omega}\cdot\mathbb{I}\vec{\Omega}}+\frac{1}{2}\cchevrons{\vec{\omega}\cdot \mathbb{I}}\cdot\vec{\omega}_0\nonumber\\
  &=\frac{1}{2}\Big[\cchevrons{m\abs{\vec{V}}^2}+\cchevrons{\vec{\Omega}\cdot\mathbb{I}\vec{\Omega}}\Big]+\frac{1}{2}\Big[m\abs{\vec{v}_0}^2+\vec{\omega_0}\cdot\overline{\mathbb{I}}\vec{\omega}_0\Big]\nonumber\\
  &=:\psi_0+\psi_K,\label{eq:energyDecomposition}
\end{align}
where we make use of the fact that both $\cchevrons{V}$ and $\cchevrons{\mathbb{I}\vec{\Omega}}$ are zero, and that $\vec{\omega}_0=\cchevrons{\mathbb{I}}^{^{-1}}\vec{\eta}$.
The first term in \eqref{eq:energyDecomposition} is the internal energy $\cchevrons{\theta}$, from now on denoted $\psi_0$ to avoid using the $\cchevrons{\cdot}$ notation.
The second term is the kinetic energy. It may be counterintuitive to identify $\psi_0(\vec{q}, t)$ as the internal energy, but from the decomposition \eqref{eq:energyDecomposition} we can obtain:
\begin{equation}
  \psi_0 = \psi-\psi_K
\end{equation}
which can be interpreted as the well-known decomposition of internal energy into total energy minus macroscopic kinetic energy. 

We will next derive a slightly different decomposition in order to relate some terms of $\psi(\vec{q},t)$ to the classical Oseen--Frank energy functional~\cite{frank,oseen}.
It is a well-known fact that the inertia tensor of a needle-shaped molecule, i.e.~the vanishing-girth limit of a calamitic molecule, can be decomposed as:
\begin{equation}
\mathbb{I}=\lambda_1 \Big(I-\vec{\nu}\otimes\vec{\nu}^T\Big)\label{eq:InertiaTensorNeedle}
\end{equation}
with $\vec{\nu}\in \mathbb{S}^2$.
Motivated by \eqref{eq:InertiaTensorNeedle} we assume the following decomposition holds for calamitic molecules:
\begin{equation}
  \mathbb{I}=\lambda_1 \Big(I-\vec{\nu}\otimes\vec{\nu}\Big)+\mathcal{O}(\varepsilon),\label{eq:TensorDecomposition}
\end{equation}
where now $\vec{\nu}(\vec{q},t)\in \mathbb{S}^2$ is the nematic vector, first introduced in \cite{ericksen2}, and $\varepsilon \approx \left(\delta/\ell\right)^2$, where $\delta$ and $\ell$ are respectively the diameter of the spherical hemisphere at the end of the calamitic molecule and the height of the calamitic molecule.
Notice that $\vec{\nu}$ is normally defined in terms of the orientation distribution function (ODF), i.e.~$\vec{\nu}=\cchevrons{\vec{\mathcal{N}}(\vec{\alpha})}$, where $\vec{\mathcal{N}}(\vec{\alpha})$ is the unit vector corresponding to the Euler angles $\vec{\alpha}$.

The introduction of a nematic director allows for substantial simplifications of computations involving the average of the inertia tensor, since under the hypothesis of the emergence of a nematic ordering we are allowed to move the inertia tensor expressed as in \eqref{eq:TensorDecomposition} outside the $\cchevrons{\cdot}$.
Since by theorem \ref{thm:totalDerivativeDirector} the material time derivative $\dot{\vec{\nu}}$ can be computed as $\vec{\omega}\times\vec{\nu}=\dot{\vec{\nu}}$, we can obtain:
\begin{align}
\cchevrons{\vec{\omega}_0\cdot \mathbb{I}\vec{\omega}_0}&=\lambda_1\Big[\vec{\omega}_0\cdot\vec{\omega}_0-\cchevrons{(\vec{\omega}_0\cdot\vec{\nu})(\vec{\omega}_0\cdot\vec{\nu})}\Big]+\mathcal{O}(\varepsilon)\\
&=\lambda_1\dot{\vec{\nu}}\cdot \dot{\vec{\nu}}{\,+\,\overset{\circ}{\vec{\nu}}\cdot \overset{\circ}{\vec{\nu}}-2(\dot{\vec{\nu}}\cdot \overset{\circ}{\vec{\nu}})}+\mathcal{O}(\varepsilon),\label{eq:EricksenEnergy}
\end{align}
using the triple product identity {$(\vec{\omega}_0\times\vec{\nu})\cdot(\vec{\omega}_0\times\vec{\nu})=(\vec{\omega}_0\cdot\vec{\omega}_0)(\vec{\nu}\cdot\vec{\nu})-(\vec{\omega}_0\cdot\vec{\nu})(\vec{\nu}\cdot\vec{\omega}_0)$, together with the definition of the corotational time derivative of the director, i.e.~$\overset{\circ}{\vec{\nu}}\coloneqq \dot{\vec{\nu}}-\vec{\omega}_0 \times \vec{\nu}$}.
This simplifies the kinetic energy
\begin{equation}
\psi_K=\frac{1}{2}m(\vec{v}_0\cdot\vec{v}_0)+\frac{\lambda_1}{2}\Big(\dot{\vec{\nu}}\cdot \dot{\vec{\nu}}+{\overset{\circ}{\vec{\nu}}\cdot \overset{\circ}{\vec{\nu}} - 2(\dot{\vec{\nu}}\cdot \overset{\circ}{\vec{\nu}})}\Big)+\mathcal{O}(\varepsilon)\label{eq:EricksenKineticEnergy}
\end{equation}
which in the vanishing girth limit, i.e.~$\varepsilon \searrow 0$, becomes $\psi_K(\vec{q},t)=\frac{1}{2}m(\vec{v}_0\cdot\vec{v}_0)+\frac{\lambda_1}{2}(\dot{\vec{\nu}}\cdot \dot{\vec{\nu}}+\overset{\circ}{\vec{\nu}}\cdot \overset{\circ}{\vec{\nu}}-2\dot{\vec{\nu}}\cdot\overset{\circ}{\vec{\nu}})$.
We proceed with our analysis of $\psi_0(\vec{q},t)$, making use of the fact that by theorem \ref{thm:totalDerivativeDirector} $\vec{\omega}\times \vec{\nu}=\partial_t \vec{\nu}+(\nabla_{\vec{q}}\vec{\nu})\vec{v}$, hence we can expand the $\cchevrons{\vec{\omega}\cdot\mathbb{I}\vec{\omega}}$ in $\psi$ as
\begin{align}
  2\lambda_1^{-1}\cchevrons{\vec{\omega}\cdot \mathbb{I}\vec{\omega}}&=\color{orange}\cchevrons{(\vec{\omega}\times\vec{\nu})\cdot (\vec{\omega}\times \vec{\nu})}=\cchevrons{\Big(\partial_t \vec{\nu}+(\nabla_{\vec{q}} \vec{\nu}) \vec{v}\Big)^T\Big(\partial_t \vec{\nu}+(\nabla_{\vec{q}} \vec{\nu})\vec{v}\Big)}\nonumber \\
  &=\cchevrons{\abs{\partial_t \vec{\nu}}}^2+2\cchevrons{(\partial_t \vec{\nu})^T\Big((\nabla_{\vec{q}}\vec{\nu})\vec{v}\Big)}+\tr\cchevrons{\nabla_{\vec{q}}{\vec{\nu}}(\vec{v}\otimes\vec{v})\nabla_{\vec{q}}\vec{\nu}^T},\nonumber
\end{align}
where the last term has been obtained by $tr(\vec{a}\vec{b}^T)=\vec{a}^T\vec{b}$.
Furthermore, we notice that in the mixed-term
\begin{equation}
\cchevrons{2(\partial_t \vec{\nu})^T\Big((\nabla_{\vec{q}}\vec{\nu})\vec{v}\Big)}=-2\cchevrons{\abs{\partial_t \vec{\nu}}^2}+\cchevrons{2\partial_t{\vec{\nu}}\cdot \dot{\vec{\nu}}},
\end{equation}
the last term vanishes because $\partial_t{\vec{\nu}}\cdot \dot{\vec{\nu}}=\partial_t\vec{\nu}\cdot (\vec{\omega}\times \vec{\nu})=\vec{\omega}\cdot (\partial_t \vec{\nu} \times \vec{\nu})=0$, as $\vec{\nu}\times \vec{\nu} = 0$.
We are now left with the following expression for the internal energy:
\begin{align}
  \psi_0& = \psi-\psi_K= \psi-\frac{1}{2}m(\vec{v}_0\cdot\vec{v}_0)-\frac{\lambda_1}{2}(\dot{\vec{\nu}}\cdot\dot{\vec{\nu}})+{\frac{\lambda_1}{2}(\overset{\circ}{\vec{\nu}}\cdot \overset{\circ}{\vec{\nu}}) - \lambda(\dot{\vec{\nu}}\cdot\overset{\circ}{\vec{\nu}})}\nonumber\\
  & =\frac{1}{2}m\cchevrons{\abs{\vec{v}}^2}+\frac{\lambda_1}{2}\tr\cchevrons{\nabla_{\vec{q}}{\vec{\nu}}(\vec{v}\otimes\vec{v})\nabla_{\vec{q}}\vec{\nu}^T}-\frac{1}{2}m(\vec{v}_0\cdot\vec{v}_0)\nonumber\\&-\frac{\lambda_1}{2}(\dot{\vec{\nu}}\cdot\dot{\vec{\nu}})-{\frac{\lambda_1}{2}}\cchevrons{(\partial_t \vec{\nu}\cdot \partial_t \vec{\nu})}+{\frac{\lambda_1}{2}(\overset{\circ}{\vec{\nu}}\cdot \overset{\circ}{\vec{\nu}})-\lambda(\dot{\vec{\nu}}\cdot\overset{\circ}{\vec{\nu}})}\nonumber\\
  & =\frac{1}{2}\cchevrons{\abs{\vec{v}^2}}+\frac{\lambda_1}{2}\tr\Big[\nabla_{\vec{q}}\vec{\nu}\cchevrons{\Pi}\nabla_{\vec{q}}\vec{\nu}^T\Big]-\frac{1}{2}m(\vec{v}_0\cdot\vec{v}_0)\nonumber\\&-\frac{\lambda_1}{2}(\dot{\vec{\nu}}\cdot\dot{\vec{\nu}})-{\frac{\lambda_1}{2}}\cchevrons{(\partial_t \vec{\nu}\cdot \partial_t \vec{\nu})}+{\frac{\lambda_1}{2}(\overset{\circ}{\vec{\nu}}\cdot \overset{\circ}{\vec{\nu}})-\lambda(\dot{\vec{\nu}}\cdot\overset{\circ}{\vec{\nu}})}\nonumber\\
  & =\frac{1}{2}\cchevrons{\abs{\vec{v}^2}}+\frac{\lambda_1}{2}\tr\Big[\nabla_{\vec{q}}\vec{\nu}\mathbb{P}\nabla_{\vec{q}}\vec{\nu}^T\Big]-\frac{1}{2}m(\vec{v}_0\cdot\vec{v}_0)+{\frac{\lambda_1}{2}(\overset{\circ}{\vec{\nu}}\cdot \overset{\circ}{\vec{\nu}})-\lambda(\dot{\vec{\nu}}\cdot\overset{\circ}{\vec{\nu}})}.
\end{align} 

To conclude this section we would like to highlight one term in the previous decomposition of $\psi_0$, i.e~\eqref{eq:energyDecomposition}, which we denote as
\begin{equation}
  \psi_{OF}(\vec{q},t)=\frac{\lambda_1}{2}\tr\Big[\nabla_{\vec{q}}\vec{\nu}\mathbb{P}\nabla_{\vec{q}}\vec{\nu}^T\Big].
\end{equation}
Using the Maxwellian to compute $\mathbb{P}$, we have $\mathbb{P}^{(0)}=\frac{2}{5m}(I_1I_2I_3)^{\frac{1}{2}}\cchevrons{\theta}I$.
Since multiples of the identity matrix commute with all other matrices, we get the one-constant approximation of the Oseen--Frank energy functional, with the caveat that the Frank constant in this case depends on the pressure:
\begin{equation}
  \label{eq:OseenFrankEnergy}
  \rho\psi_{OF}^{(0)}(\vec{q},t)=p_K\frac{\lambda_1}{2}tr\Big[\nabla_{\vec{q}}\vec{\nu}\nabla_{\vec{q}}\vec{\nu}^T\Big].
\end{equation}
The dependence of the energy functional on the pressure is a striking difference from the classical Oseen--Frank theory.
Such a dependence is a consequence of the fact that the principle of virtual work we started from and the constitutive relation for the pressure tensor we used were derived from the Boltzmann--Curtiss equation, which describes the dynamics of compressible fluids.

\section{Noll--Coleman procedure}
\label{sec:NollColeman}
We observe that \eqref{eq:equationOfPowerExpanded} holds for all thermokinetic processes to which the thermodynamic body representing our fluid can be subjected, in the sense of~\cite{truesdell2}.
We can therefore perform a Noll--Coleman procedure to derive from $\psi^{(0)}_{OF}$ the nematic contribution to the stress tensor and the couple stress tensor.
The Noll--Coleman procedure is a method to derive the constitutive equations for the moments from a principle of virtual work \cite{nollColeman,stewart}.

First, we recall \eqref{eq:equationOfPowerExpanded}, where we have removed the dependence on the density:
\begin{align}
  \dot{\psi}_0+m\mathbb{P}:\nabla_{\vec{q}}\vec{v}_0+\mathbb{M}:\nabla_{\vec{q}}\vec{\omega}_0+\nabla_{\vec{q}}\cdot\vec{\mathfrak{Q}}= 0.\label{eq:equationOfPowerExpandedWithOutDivergence}
\end{align}

We are interested in applying the Noll--Coleman procedure only to find the relation between the pressure tensor and the nematic director, so we fix $\psi_0 =\psi_{OF}$. To begin with, we expand the total derivative of $\psi_0$ in terms of $\dot{\vec{\nu}}$ and $\dot{\nabla_{\vec{q}}\vec{\nu}}$ and then adopt tensor notation. We then proceed using the identity $\dot{\vec{\nu}}=\vec{\omega}_0\times\vec{\nu}$ and definition of material time derivative to obtain $\dot{\nabla_{\vec{q}}\vec{\nu}}=\nabla_{\vec{q}}\dot{\vec{\nu}}-\nabla_{\vec{q}}\vec{\nu}\nabla_{\vec{q}}\vec{v}_0$, which we substitute into the expansion of $\dot{\psi}_0$:
\begin{align}
  \dot{\psi_0}&=\frac{\partial\psi_0}{\partial\vec{\nu}}\dot{\vec{\nu}}+\frac{\partial\psi_0}{\partial(\nabla_{\vec{q}}\vec{\nu})}\dot{\nabla_{\vec{q}}\vec{\nu}}=\frac{\partial\psi_0}{\partial(\nu_p)}\Big[\varepsilon_{iqp}\nu_q\omega_{0,i}\Big]+\frac{\partial\psi_0}{\partial(\partial_k\nu_p)}\partial_k\dot{\nu}_p\nonumber\\
  &=\frac{\partial\psi_0}{\partial(\nu_p)}\Big[\varepsilon_{iqp}\nu_q\omega_{0,i}\Big]+\frac{\partial\psi_0}{\partial(\partial_k\nu_p)}\Big[\varepsilon_{iqp}\partial_{k}(\nu_q\omega_{0,i})-\partial_q(\nu_p)\partial_k(v_{0,q})\Big]\nonumber\\
  &=\frac{\partial\psi_0}{\partial(\nu_p)}\Big[\varepsilon_{iqp}\nu_q\omega_{0,i}\Big]+\frac{\partial\psi_0}{\partial(\partial_k\nu_p)}\Big[\varepsilon_{iqp}\partial_k(\nu_q)\omega_{0,i}+\varepsilon_{iqp}\partial_k(\omega_{0,i})\nu_q-\partial_q(\nu_p)\partial_k(v_{0,q})\Big]\nonumber\\
  &=\varepsilon_{iqp}\Bigg[\Big(\nu_q \frac{\partial\psi_0}{\partial(\nu_p)}+\partial_k\nu_q\frac{\partial\psi_0}{\partial(\partial_k\nu_p)}\Big)\omega_{0,i}+\nu_q \frac{\partial\psi_0}{\partial(\partial_k\nu_p)}\partial_k\omega_{0,i}\Bigg]-\frac{\partial\psi}{\partial(\partial_k\nu_p)}\partial_q\nu_p\partial_k(v_{0,q})\label{eq:preEricksenInequality}.
\end{align}
We would like to express the previous equation using a vector product. To do this we use the Ericksen identity stated in the following lemma~\cite[Appendix B]{stewart}.
\begin{lemma}
  \label{lem:EricksenIdentity}
  Given a rotationally invariant $\psi_0=\psi_0(\vec{\nu},\nabla_{\vec{q}}\vec{\nu})$, so that $\psi_0(\vec{\nu},\nabla_{\vec{q}}\vec{\nu})\allowbreak =\psi_0(Q\vec{\nu},Q\nabla_{\vec{q}}Q^T)$ for all $Q\in SO(\mathbb{R}^3)$,
  the following identity holds:
  \begin{equation}
    \varepsilon_{iqp}\Big[\nu_q\frac{\partial\psi_0}{\partial\nu_p}+\partial_k\nu_q\frac{\partial\psi_0}{\partial(\partial_k\nu_p)}+\partial_q\nu_k\frac{\partial\psi_0}{\partial(\partial_p\nu_k)}\Big]=0.\label{eq:EricksenIdentity}
  \end{equation}
\end{lemma}
Using \eqref{eq:EricksenIdentity} to substitute for the coefficients of $\omega_{0,i}$ in \eqref{eq:preEricksenInequality} yields
\begin{equation}
  \dot{\psi_0}=\varepsilon_{iqp}\Big(\nu_q \frac{\partial\psi_0}{\partial(\partial_k\nu_p)}\partial_k\omega_{0,i}-\partial_q\nu_k\frac{\partial\psi_0}{\partial(\partial_p\nu_k)}\omega_{0,i}\Big)-\frac{\partial\psi_0}{\partial(\partial_k\nu_p)}\partial_q\nu_p\partial_k(v_{0,q})\label{eq:postEricksenInequality}.
\end{equation}
Substituting \eqref{eq:postEricksenInequality} into \eqref{eq:equationOfPowerExpandedWithOutDivergence} and overloading the cross product between a vector and a tensor as $\vec{\nu}\times S = \varepsilon_{iqp}\nu_qS_{jp}$ we get, for a thermokinetic process with $\vec{\omega}_{0} = 0$, the following identity
\begin{align}
  \Big(\mathbb{P}-(\frac{\partial\psi_0}{\partial(\nabla\vec{\nu})})^T\nabla_{\vec{q}}\vec{\nu}\Big):\nabla_{\vec{q}}\vec{v}_0+\Big(\mathbb{M}+\vec{\nu}\times(\frac{\partial\psi_0}{\partial\nabla\vec{\nu}})\Big):\nabla_{\vec{q}}\vec{\omega}_0+\nabla_{\vec{q}}\cdot \vec{\mathfrak{Q}}= 0.
\end{align}
{
  Under the assumption that we are sufficiently close to the thermodynamic equilibrium, we can compute the heat flux $\nabla_{\vec{q}}\cdot\vec{\mathfrak{Q}}$ using the Maxwellian distribution \eqref{eq:Maxwellian}, i.e.~ $\nabla_{\vec{q}}\cdot\vec{\mathfrak{Q}^{(0)}}=0$.
  Hence, near thermodynamic equilibrium the previous equation can be approximated as
  \begin{align}
    \Big(\mathbb{P}-(\frac{\partial\psi_0}{\partial(\nabla\vec{\nu})})^T\nabla_{\vec{q}}\vec{\nu}\Big):\nabla_{\vec{q}}\vec{v}_0+\Big(\mathbb{M}+\vec{\nu}\times(\frac{\partial\psi_0}{\partial\nabla\vec{\nu}})\Big):\nabla_{\vec{q}}\vec{\omega}_0= 0\label{eq:equationOfPowerExpandedMaxwellian}.
  \end{align}
Furthermore, since we can choose $\nabla_{\vec{q}}\vec{v}_0$ and $\nabla_{\vec{q}}{\vec{\omega}_0}$ arbitrarily to preserve the equality in \eqref{eq:equationOfPowerExpandedMaxwellian}, the following identities must always hold:
\begin{equation}
  \mathbb{P}^{(*)}=\left(\frac{\partial\psi_{OF}}{\partial(\nabla\vec{\nu})}\right)^T\nabla_{\vec{q}}\vec{\nu}, \qquad \mathbb{M}^{(*)}=-\vec{\nu}\times\left(\frac{\partial\psi_{OF}}{\partial\nabla\vec{\nu}}\right),\label{eq:constitutiveNollColemann}
\end{equation}
}
where the superscript $^{(*)}$ denotes the additional component that arises from the nematic ordering, via the Noll--Coleman procedure.
This step is crucial: it yields a closure that describes the coupling between the pressure stress tensor and $\nabla\vec{\nu}$.
Therefore, in contrast to the hypothesis of Curtis and Dahler \cite{curtissI,curtissV}, the intrinsic angular momentum is coupled to the macroscopic velocity.

Furthermore, \eqref{eq:lawAngularMomentum} can be used to describe the evolution of the nematic director.
We proceed to substitute \eqref{eq:constitutiveNollColemann} inside \eqref{eq:continuumMechanics} to obtain the following set of equations, governing the dynamics of a polyatomic rarefied gas:
\begin{subequations}
  \begin{gather}
    \partial_t \rho + \nabla_{\vec{q}}\cdot (\rho\vec{v}_0) =0,\label{eq:continuityEquationEta}\\
    \rho\Big[\partial_t \vec{v}_0 + (\nabla_{\vec{q}}\vec{v}_0)\vec{v}_0\Big] + \nabla_{\vec{q}}\cdot\left(\rho\,\mathbb{P}+\rho\left(\frac{\partial\psi_{OF}}{\partial\nabla\vec{\nu}}\right)^T\nabla_{\vec{q}}\vec{\nu}\right)=0,\label{eq:lawLinearMomentumEta}\\
    \rho\Big[\partial_t \vec{\eta} + (\nabla_{\vec{q}}\vec{\eta})\vec{v}_0\Big] -\nabla_{\vec{q}}\cdot\left(\rho\,\vec{\nu}\times\left(\frac{\partial\psi_{OF}}{\partial\nabla\vec{\nu}}\right)\right) = 0.\label{eq:lawAngularMomentumEta}
  \end{gather}
\end{subequations}
We focus our attention on the equation for the conservation of angular momentum \eqref{eq:lawAngularMomentumEta} and notice that we have two variables $\vec{\nu}$ and $\vec{\eta}$ but only one equation.
Ideally, we would like to make one of the two disappear or add an equation relating the two quantities.
In deriving the Leslie--Ericksen equations (see for example~\cite{stewart}), Leslie wrote~\cite{leslie2}:
\begin{quote}
  The inertial term associated with local rotation of the material element is omitted because in general, it is negligible.
\end{quote}
Given the symmetric nature of the stress tensor, in our model we can not ignore the intrinsic angular momentum $\vec{\eta}$, since it is the only source of anisotropy. In order to express the intrinsic angular momentum $\vec{\eta}$ in terms of the nematic director $\vec{\nu}$ we observe that if as in \eqref{eq:TensorDecomposition} we assume $\varepsilon\searrow 0$ and use the two identities
\begin{equation}
\dot{\vec{\nu}}=\vec{\omega}\times\vec{\nu}
\end{equation}
and
\begin{equation}
\dot{\vec{\nu}}\times\vec{\nu} = (\vec{\omega}\times \vec{\nu})\times \vec{\nu} = \Big(\vec{\omega}(\vec{\nu}^T\vec{\nu})-\vec{\nu}(\vec{\nu}^T\vec{\omega})\Big),
\end{equation}
we have:
\begin{equation}
    \vec{\eta} = \cchevrons{\mathbb{I}\vec{\omega}} =\cchevrons{\lambda_1\Big(I-\vec{\nu}\otimes\vec{\nu}\Big)\vec{\omega}}=\cchevrons{\lambda_1\Big(\vec{\omega}(\vec{\nu}^T\vec{\nu})-\vec{\nu}(\vec{\nu}^T\vec{\omega})\Big)}=\lambda_1\dot{\vec{\nu}}\times \vec{\nu}.
\end{equation}
Substituting this last equation into \eqref{eq:lawAngularMomentumEta} we obtain
\begin{equation}
    \rho{\ddot{\vec{\nu}}}\times\vec{\nu}-\nabla_{\vec{q}}\cdot\left(\rho\,\vec{\nu}\times\left(\frac{\partial\psi_{OF}}{\partial\nabla\vec{\nu}}\right)\right)  =  0.\nonumber
\end{equation}
To massage the previous equation to a better-known form we subtract from the previous equation the antisymmetric part of $\rho(\frac{\partial\psi_{OF}}{\partial(\nabla\vec{\nu})})^T\nabla_{\vec{q}}\vec{\nu}$, which we know will vanish since $\mathbb{P}$ is symmetric by definition.
Using the Ericksen identity once again yields, in tensor notation,
\begin{gather}
  \varepsilon_{ikq}\Big[\lambda_1{\ddot{\nu}_q}\nu_k+\nu_k\partial_j \Big(\rho\frac{\partial \psi_{OF}}{\partial(\partial_j\nu_q)}\Big)-\rho\frac{\partial\psi_{OF}}{\partial(\partial_k\nu_p)}\partial_q\nu_p\Big]\nonumber\\
  =\varepsilon_{ikq}\nu_k\Big[\lambda_1{\ddot{\nu}_q}+\partial_j \Big(\rho\frac{\partial \psi_{OF}}{\partial(\partial_j\nu_q)}\Big)-\rho\frac{\partial\psi_{OF}}{\partial \nu_q}\Big]  =  0.
\end{gather}
We conclude by observing that the last equation prescribes that all the components of
\begin{equation}
\lambda_1\rho{\ddot{\vec{\nu}}}+\nabla_{\vec{q}}\cdot \rho\frac{\partial\psi_{OF}}{\partial\nabla\vec{\nu}}-\rho\frac{\partial\psi_{OF}}{\partial\vec{\nu}},
\end{equation}
in the direction orthogonal to $\vec{\nu}$ vanish. We can express this constraint using the Lagrange multiplier $\tau(\vec{q}):\mathbb{R}^3\to\mathbb{R}$, i.e.~
\begin{equation}
  \lambda_1\rho{\ddot{\vec{\nu}}}+\nabla_{\vec{q}}\cdot \rho\frac{\partial\psi_{OF}}{\partial\nabla\vec{\nu}}-\rho\frac{\partial\psi_{OF}}{\partial\vec{\nu}}=\tau\vec{\nu},
\end{equation}
where the Lagrange multiplier $\tau(\vec{q})$ is determined by the unit length constraint on $\vec{\nu}$.
We have retrieved the balance law of angular momentum presented in \cite{leslie2}.

Next we evaluate the expression for $\frac{\partial \psi_{OF}}{\partial\nabla\vec{\nu}}$ using the Maxwellian Oseen--Frank energy $\psi_{OF}^{(0)}$, using the Maxwellian pressure tensor $\mathbb{P}^{(0)}$ to obtain the following set of equations:
\begin{subequations}
  \label{eq:LeslieEricksenWithPressure}
  \begin{align}
    \partial_t \rho + \nabla_{\vec{q}}\cdot (\rho\vec{v}_0) &=0,\\
    \rho\Big[\partial_t \vec{v}_0 + (\nabla_{\vec{q}}\vec{v}_0)\vec{v}_0\Big] + \nabla_{\vec{q}}\cdot\Big(p_{K}I+p_{K}\frac{\lambda_1}{2}\nabla_{\vec{q}}\vec{\nu}^T\nabla_{\vec{q}}\vec{\nu}\Big)&=0,\\
    \lambda_1\rho {\ddot{\vec{\nu}}}+\nabla_{\vec{q}}\cdot\Big(p_K\frac{\lambda_1}{2}\nabla_{\vec{q}}\vec{\nu}\Big) &= \tau\vec{\nu},\\
    \norm{\vec{\nu}} &= 1,
  \end{align} 
\end{subequations}
where the last equation is the constraint on the nematic director length that comes from \eqref{eq:InertiaTensorNeedle}.
We observe that the previous system of equations has five unknowns ($\rho$, $\vec{v}_0$, $p_K$, $\lambda$, $\vec{\nu}$) and only four equations.
This is because we still need to consider the equation describing the evolution of $\psi_0$.
To obtain this equation we start from \eqref{eq:equationOfPowerExpandedWithOutDivergence}, with adiabatic conditions, 
\begin{equation}
  \rho\Big[\partial_t\psi_0+\nabla_{\vec{q}}\psi_0\cdot \vec{v}_0\Big]+\Big(p_{K}I+p_{K}\frac{\lambda_1}{2}\nabla_{\vec{q}}\vec{\nu}^T\nabla_{\vec{q}}\vec{\nu}\Big):\nabla_{\vec{q}}\vec{v}_0 = 0,\label{eq:LeslieEricksenWithPressureEnergy}
\end{equation}
and $\psi_0$ can be related to the other five variables in the system via \eqref{eq:kineticPressure}.
We conclude now by observing that the combination of \eqref{eq:LeslieEricksenWithPressure}, \eqref{eq:LeslieEricksenWithPressureEnergy} and \eqref{eq:kineticPressure} form a closed system that describes the evolution of a rarefied gas subject to a nematic ordering, near thermodynamic equilibrium and in adiabatic conditions.
For completeness, we report the final system of six equations determining the evolution of the six macroscopic unknowns $\rho$, $\vec{v}_0$, $p_K$, $\vec{\nu}, \tau$ and $\psi_0$:
\begin{subequations}
  \label{eq:LeslieEricksenWithPressureAndEnergy}
  \begin{align}
    \partial_t \rho + \nabla_{\vec{q}}\cdot (\rho\vec{v}_0) &= 0,\\
    \rho\Big[\partial_t \vec{v}_0 + (\nabla_{\vec{q}}\vec{v}_0)\vec{v}_0\Big] + \nabla_{\vec{q}}\cdot\Big(p_{K}I+p_{K}\frac{\lambda_1}{2}\nabla_{\vec{q}}\vec{\nu}^T\nabla_{\vec{q}}\vec{\nu}\Big) &= 0,\\
    \lambda_1\rho {\ddot{\vec{\nu}}}+\nabla_{\vec{q}}\cdot\Big(p_K\frac{\lambda_1}{2}\nabla_{\vec{q}}\vec{\nu}\Big) &= \tau\vec{\nu},\\
    \rho\Big[\partial_t\psi_0+\nabla_{\vec{q}}\psi_0\cdot \vec{v}_0\Big]+\Big(p_{K}I+p_{K}\frac{\lambda_1}{2}\nabla_{\vec{q}}\vec{\nu}^T\nabla_{\vec{q}}\vec{\nu}\Big):\nabla_{\vec{q}}\vec{v}_0 &= 0,\\
    p_K &= \frac{6}{5}\frac{\rho}{m}(I_1I_2I_3)^{\frac{1}{2}}\psi_0, \\
    \norm{\vec{\nu}} &= 1.
  \end{align} 
\end{subequations}
These derived equations \eqref{eq:LeslieEricksenWithPressureAndEnergy} are an inviscid compressible variant of the Leslie--Ericksen equations. Under isothermal and isobaric conditions, this system reduces to the inviscid Leslie--Ericksen equations.
\section{Conclusions}
\; In this work we have derived a set of equations describing the evolution of a rarefied gas subject to a nematic ordering near thermodynamic equilibrium, using a novel moment closure technique inspired by the Noll--Coleman procedure.
At the heart of the proposed moment closure technique remains the Chapman--Enskog expansion, which approximates the distribution function as a perturbation from the Maxwellian with order corresponding to increasing powers of the Knudsen number.
We considered the first-order approximation, which is only valid near thermodynamic equilibrium, i.e.~for extremely small Knudsen numbers.
The derived equations are an inviscid compressible variant of the Leslie--Ericksen equations, which describe the evolution of the density, the velocity, the pressure, the nematic director, the Lagrange multiplier enforcing the unit length constraint on the nematic director, and the internal energy.

We remark that the main differences between the derived equations and the classical Leslie--Ericksen equations are the compressibility of the fluid and the lack of viscosity.
This last point is a consequence of the well-known fact that the first-order Chapman--Enskog expansion does not capture viscous effects.

In future work, we plan to extend our derivation using a higher-order Chapman--Enskog expansions or the van Kampen elimination procedure to include viscous effects \cite{vanKampen}.
We hope such an extension will allow us to derive a viscous compressible variant of the Leslie--Ericksen equations.

Lastly, we plan to numerically investigate the behavior of the derived equations and their viscous extension, to better understand the role of the nematic ordering in the dynamics of rarefied gases and to compare the results with the classical Leslie--Ericksen equations.

\appendix
\section{Rational mechanics of rigid bodies}
\label{sec:rationalMechanics}
Here we give a brief summary of the basic notions of the rational mechanics of rigid bodies.
In particular, we will discuss our choice of generalized coordinates used to describe the configuration space of a rigid body and what are the corresponding conjugate momenta.
We will assume our calamitic molecules can be modeled as a discrete rigid system of $N$ material points and limit our focus to this specific case. It is trivial to extend the notions here presented to the case $N\to\infty$.
The presentation is primarily drawn from \cite{fasanoMarmi}.
We redirect the reader interested in a more in-depth study of this topic to Goldstein et al.~or Whittaker \cite{goldsteinPooleSafko,whittaker}.
\begin{definition}
  We call a {rigid body} a set of material points $\{(P_1,m_1)\}_{i=1}^N$, that at all time satisfy a set of so-called {rigidity constraints}:
  \begin{equation}
    \norm{P_i-P_j}=R_{ij},\qquad 1\leq i\leq j \leq N. 
  \end{equation}
  where $\norm{\cdot}$ denotes the standard Euclidean norm in $\mathbb{R}^3$.
  Furthermore, we require the rigidity constraints to satisfy the Euclidean compatibility conditions, i.e.~
  \begin{enumerate}
    \item $R_{ij}\geq 0$ and $R_{ij}=0$ if and only if $P_i=P_j$,
    \item if $P_{i'} =  \lambda P_i$ and $P_{j'}=\lambda P_j$ then $R_{i'j'}=\lambda R_{ij}$,
    \item $R_{ij}\leq R_{ik}+R_{kj}$ for all $1\leq k < j \leq N$. 
  \end{enumerate}
  Lastly, we require that if $N>3$, not all $\frac{N(N-1)}{2}$ equations are independent. In particular we require the system of the rigidity equations to have rank at most $3$.
\end{definition}
The system of rigidity constraints having at most rank 3 implies that once the position of three points are determined, all other points of the material body can be located using the rigidity constraints.

If we have a rigid body consisting of three unaligned material points, $P_1$, $P_2$, and $P_3$, and a fixed reference frame $(O,\vec{e}_1,\vec{e}_2,\vec{e}_3)$ we can identify the position of the three points assigning three coordinates to $P_1$, two coordinates to $P_2$ and one coordinate to $P_3$.
Therefore, any three unaligned points can be identified by six coordinates.
Notice also that this tells us that any rigid body can be identified by six coordinates provided that it contains at least three unaligned points. In the case where all points in the system are aligned, we lose one degree of freedom because we only need five coordinates to identify two key material points and then all the other material points of the rigid body can be located using the rigidity constraints.
\begin{proposition}
  The configuration space for a rigid body containing at least three unaligned points is $\mathbb{R}^3\times SO(3)$.
\end{proposition}
Since we can represent an element of $SO(3)$ via an orthogonal matrix with positive determinant, we can identify the configuration space of a rigid body with the position of its center of mass $O$, and the three vectors $\vec{e}_1,\,\vec{e}_2,\,\vec{e}_3$ representing the column of the orthogonal matrix.
We call the quadruplet $(O,\vec{e}_1,\vec{e}_2,\vec{e}_3)$ a body \emph{reference frame}.

We now introduce Euler angles. Given two different body reference frames $(O,\vec{e}_1,\allowbreak \vec{e}_2,\vec{e}_3)$ and $(O,\hat{\vec{e}_1},\hat{\vec{e}_2},\hat{\vec{e}}_3)$ we can represent the $SO(3)$ transformation that rotates one reference frame to the other making use of three angles $\alpha_1$, $\alpha_2$ and $\alpha_3$
(see Fig.~\ref{fig:EulerAngles}):
\begin{enumerate}
  \item the \emph{precession angle} $\alpha_1$, defined as the angle around $\vec{e}_3$ that aligns $\vec{e}_1$ with the intersection of the planes identified by the normals $\vec{e}_3$ and $\hat{\vec{e}}_3$;
  \item the \emph{nutation angle} $\alpha_2$, defined as the angle between $\vec{e}_3$ and $\hat{\vec{e}}_3$;
  \item the \emph{intrinsic rotation angle} $\alpha_3$, defined as the angle around $\hat{\vec{e}}_3$ that aligns $\hat{\vec{e}}_1$ with the intersection of the planes identified by the normals $\vec{e}_3$ and $\hat{\vec{e}}_3$.  
\end{enumerate}
\begin{figure}[h]
    \caption{A depiction of the precession, nutation and intrinsic rotation angles.}
    \centering
    \vspace{0.5cm}
    \label{fig:EulerAngles}
    \scalebox{0.5}[0.5]{\begin{tikzpicture}[line cap=round]
\def\pre{-30} 
\def\nut{40}  
\def\rop{30}  
\begin{scope}[shift={(-8,0)},my view]
\draw[fill=my cyan] (0,0) circle (2);
\draw[axis] (0,0,0) -- (4,0,0)     node[below] {$\hat{\vec{e}}_1$};
\draw[axis] (0,0,0) -- (0,4,0)     node[right] {$\hat{\vec{e}}_2$};
\draw[axis] (0,0,0) -- (0,0,4)     node[above] {$\hat{\vec{e}}_3=\vec{e}_3$};
\draw[axis] (0,0,0) -- (\pre:4)    node[below] {$\vec{e}_1$};
\draw[axis] (0,0,0) -- (90+\pre:4) node[below] {$\vec{e}_2$};
\foreach\i in {0,90}
  \draw[-latex] (\i+\pre:2.5) arc (\i+\pre:\i:2.5)
        node at (\i+0.5*\pre:3) {$\alpha_1$};
\draw[canvas is xy plane at z=2,->] (90:0.5) arc (90:360:0.5) node[pos=0,right] {Precession $\alpha_1$};
\end{scope}
\begin{scope}[my view]
\draw[fill=my magenta,nutation] (2,0) arc (0:-180:2) -- cycle;
\draw[fill=my cyan] (0,0) circle (2);
\draw[axis]   (0,0,0) -- (0,4,0) node[right] {$\vec{e}_2$};
\begin{scope}[nutation]
  \draw[fill=my magenta] (2,0) arc (0:180:2) -- cycle;
  \draw[axis] (0,0,0) -- (0,4,0) node[above] {$\hat{\vec{e}_2}$};
  \draw[axis] (0,0,0) -- (0,0,4) node[above] {$\hat{\vec{e}_3}$};
\end{scope}
\foreach\i in {0,90}
  \draw[my view,canvas is yz plane at x=0,-latex] (\i:2.5) arc (\i:\i+\nut:2.5)
                                          node at (\i+0.5*\nut:3) {$\alpha_2$};
\draw[axis] (0,0,0) -- (4,0,0) node[below] {$\vec{e}_1=\hat{\vec{e}}_1$};
\draw[axis] (0,0,0) -- (0,0,4) node[above] {$\vec{e}_3$};
\draw[canvas is yz plane at x=2,->] (-90:0.5) arc (-90:170:0.5) node[yshift=-0.3cm, pos=0.5cm,right] {Nutation $\alpha_2$};
\end{scope}
\begin{scope}[shift={(8,0)},my view,nutation]
\draw[fill=my magenta] (0,0) circle (2);
\draw[axis] (0,0,0) -- (4,0,0)     node[right] {$\hat{\vec{e}}_1$};
\draw[axis] (0,0,0) -- (0,4,0)     node[above] {$\hat{\vec{e}}_2$};
\draw[axis] (0,0,0) -- (0,0,4)     node[above] {$\vec{e}_3=\hat{\vec{e}}_3$};
\draw[axis] (0,0,0) -- (\rop:4)    node[right] {${\vec{e}}_1$};
\draw[axis] (0,0,0) -- (90+\rop:4) node[above] {${\vec{e}}_2$};
\foreach\i in {0,90}
  \draw[-latex] (\i:2.5) arc (\i:\i+\rop:2.5)
        node at (\i+0.5*\rop:3) {$\alpha_3$};
\draw[canvas is xy plane at z=2,->] (-90:0.5) arc (-90:160:0.5) node[xshift=-7mm, yshift=-7mm,pos=0,text width=2cm] {Intrinsic Rotation $\alpha_3$};
\end{scope}
\end{tikzpicture}}
\end{figure}
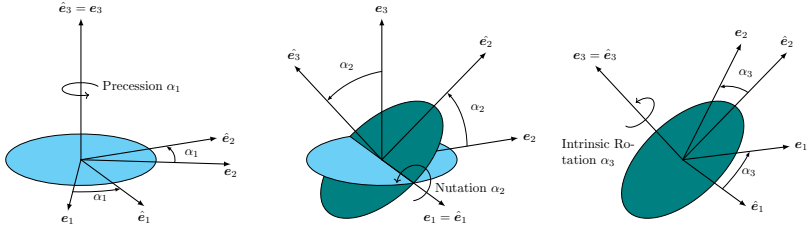
The intersection of the planes identified by the normals $\vec{e}_3$ and $\hat{\vec{e}}_3$ is called the \emph{node line} $\vec{N}$ and we denote $\vec{\alpha}=(\alpha_1,\alpha_2,\alpha_3)$.

We now proceed with some basic results.
\begin{proposition}
  The Eulerian velocity field of a rigid body can be expressed by the formula
  \begin{equation}
    v(\vec{x},t)= \vec{\omega}(t)\times \Big(\vec{x}-\vec{c}(t)\Big)+\dot{\vec{c}}(t), \label{eq:Eulerian_velocity}
  \end{equation}
  where $\vec{\omega}$ is referred to as the \emph{angular velocity} and $c(t)$ is the velocity of the center of mass of the rigid body.
\end{proposition}

\begin{theorem}\label{thm:totalDerivativeDirector}
  Under any rigid motion, we can express the total derivative of any unit vector $\vec{\nu}\in \mathbb{S}^2$, that expresses the direction connecting two points making up the rigid body, as
  \begin{equation}
    \dot{\vec{\nu}} = \vec{\omega}\times \vec{\nu}.
  \end{equation} 
\end{theorem}

\begin{proof}
Let $\vec{x}_1$ and $\vec{x}_2$ denote two points of the rigid body aligned along the direction $\vec{\nu}$, i.e.~\ 
$\vec{\nu} = (\vec{x}_2-\vec{x}_1)/\norm{\vec{x}_2-\vec{x}_1}$.
The velocity of the points can be expressed by relation 
\eqref{eq:Eulerian_velocity}:
\begin{align}
    \dot{\vec{x}}_1 =\vec{\omega}\times \Big(\vec{x}_1-\vec{c}\Big)+\dot{\vec{c}},\quad 
    \dot{\vec{x}}_2 =\vec{\omega}\times \Big(\vec{x}_2-\vec{c}\Big)+\dot{\vec{c}}.
\end{align}
Subtracting the first equation from the second we obtain:
 \begin{equation*}
     \dot{\vec{x}}_2 - \dot{\vec{x}}_1  =\vec{\omega}\times 
     \Big(\vec{x}_2-\vec{x}_1\Big)
 \end{equation*}
The result follows on division by $\norm{\vec{x}_2-\vec{x}_1}$, which remains constant under rigid motion.
\end{proof}

We now would like to express $\vec{\omega}$ as the total derivative of the Euler angles. To achieve this we decompose the angular velocity in terms of the rotations mandated by the Euler angles,
\begin{equation}
  \vec{\omega} = \dot{\alpha_1}\vec{z}+ \dot{\alpha_2} \vec{N}+\dot{\alpha_3}\hat{\vec{z}},
\end{equation}
where $(O,\vec{x},\vec{y},\vec{z})$ and $(O,\hat{\vec{x}},\hat{\vec{y}},\hat{\vec{z}})$ are the body reference frame before and after the rotation respectively, and $\vec{N}$ is the node line.
Now we introduce $\vec{N}_{\bot}=\vec{\hat{z}}\times\vec{N}$ and observe that the following identities hold:
\begin{align}
  \vec{z} &= \sin(\alpha_2)\vec{N}_{\bot}+\cos(\alpha_2)\hat{\vec{z}}, \\
  \vec{N} &= \cos(\alpha_3)\hat{\vec{x}}-\sin(\alpha_3)\hat{\vec{y}}, \\
  \vec{N}_{\bot} &= \sin(\alpha_3)\hat{\vec{x}}+\cos(\alpha_3)\hat{\vec{y}}.
\end{align}
Projecting $\vec{\omega}$ onto $(O,\hat{\vec{x}},\hat{\vec{y}},\hat{\vec{z}})$ yields
\begin{align}
  \vec{\omega} &= \Big(\dot{\alpha_2}\cos(\alpha_3)+\dot{\alpha_1}\sin(\alpha_2)\sin(\alpha_3)\Big)\hat{\vec{x}}\\
               &+ \Big(-\dot{\alpha_2}\sin(\alpha_3)+\dot{\alpha_1}\sin(\alpha_2)\cos(\alpha_3)\Big)\hat{\vec{y}} + \Big(\dot{\alpha_3}+\dot{\alpha_1}\cos(\alpha_2)\Big)\hat{\vec{z}}.\nonumber
\end{align}
We can rewrite the previous expression in matrix form as follows,
\begin{equation}
  \vec{\omega} = \begin{bmatrix}
    \sin(\alpha_2)\sin(\alpha_3) & \cos(\alpha_3) & 0\\
    \sin(\alpha_2)\cos(\alpha_3) & -\sin(\alpha_3) & 0\\
    \cos(\alpha_2) & 0 & 1
  \end{bmatrix}\begin{bmatrix}
    \dot{\alpha_1}\\\dot{\alpha_2}\\\dot{\alpha_3}
  \end{bmatrix} = \Xi\dot{\vec{\alpha}}.\label{eq:angularVelocity}
\end{equation}
We now introduce notions related to the dynamic aspects of rigid bodies, rather than the kinematic aspects.
\begin{definition}
  We define the \emph{center of mass} of a rigid body as 
  \begin{equation}
    O_{M} = \frac{1}{M}\sum_{i=1}^N m_i P_i,
  \end{equation}
  where $M$ is the total mass of the rigid body, i.e.~$M=\sum_{i=1}^{N}m_i$.
\end{definition}
Once the center of mass has been defined we can introduce the notion of \emph{inertia moment} with respect to the planes identified respectively by the orthogonal normals $\vec{n}$ and $\vec{n}'$, i.e.~
\begin{equation}
  \Gamma_{\vec{n}\vec{n}'} = \sum_{i=1}^{N}m_i \Big[(P_i-O_M)\cdot \vec{n}\Big] \Big[(P_i-O_M)\cdot \vec{n}'\Big]. 
\end{equation} 
Once we are equipped with this definition we can introduce the notion of inertia tensor with respect to a  reference frame $(O,\vec{x},\vec{y},\vec{z})$, i.e.~
\begin{equation}
  \mathbb{I} = \begin{bmatrix}
    \Gamma_{\vec{x}\vec{x}} &\Gamma_{\vec{x}\vec{y}} & \Gamma_{\vec{x}\vec{z}}\\
    \Gamma_{\vec{y}\vec{x}} &\Gamma_{\vec{y}\vec{y}} & \Gamma_{\vec{y}\vec{z}}\\
    \Gamma_{\vec{z}\vec{x}} &\Gamma_{\vec{z}\vec{y}} & \Gamma_{\vec{z}\vec{z}}
  \end{bmatrix}.
\end{equation}
In particular, we notice that since $\mathbb{I}$ is a symmetric matrix it can be orthogonally diagonalised, i.e.~
\begin{equation}
  \mathbb{I} = Q\begin{bmatrix}
    \Gamma_1 &       &       \\
        & \Gamma_2   &       \\
        &       & \Gamma_3 
  \end{bmatrix} Q^T, \qquad Q\in SO(3).
\end{equation} 
Notice that $\Gamma_i$ are positive because they can be expressed as $\Gamma_{\vec{n}\vec{n}}$ for a certain $\vec{n}$.
We will refer to $\Gamma_1$, $\Gamma_2$ and $\Gamma_3$ as the \emph{principal moments of inertia}.
We call the principal reference frame the reference frame obtained by applying to $(O,\vec{x},\vec{y},\vec{z})$ the rotation $Q$.

Lastly, we would like to apply the Hamiltonian formalism to the study of rigid motion. To achieve this we first observe that the Lagrangian corresponding to a rigid motion is only composed by the kinetic energy term,
\begin{equation}
  T = \frac{1}{2}\abs{\vec{v}}^2+\frac{1}{2}\vec{\omega}\cdot \mathbb{I}\vec{\omega}.
\end{equation}
We chose as Lagrangian coordinates the position in space and the Euler angles. We therefore need to express the above kinetic energy in terms of the total derivatives of the position and of the Euler angles. To achieve this we resort to \eqref{eq:angularVelocity},
\begin{equation}
  \mathcal{L}(\vec{q},\vec{\alpha},\dot{\vec{q}},\dot{\vec{\alpha}})=\frac{1}{2}m\abs{\dot{\vec{q}}}^2+\frac{1}{2}\dot{\vec{\alpha}}\cdot \Xi^T\mathbb{I}\Xi\dot{\vec{\alpha}}.
\end{equation}
By Sylvester's law of inertia \cite{sylvester}, the congruent transformation $D_{\dot{\vec{\alpha}}}^2\mathcal{L} = \Xi^T\mathbb{I}\Xi$
preserves the sign of the eigenvalues, i.e.~the eigenvalues of $\Xi^T\mathbb{I}\Xi$ have the same signs as the eigenvalues of $\mathbb{I}$, which are all positive. The Hessian $D^2 \mathcal{L}$ of the Lagrangian $\mathcal{L}$ expressed in terms of position, Euler angles and their total time derivative therefore remains positive definite. Thus the Legendre transformation that defines the Hamiltonian coordinates is well-posed:
\begin{equation}
  \vec{p}= \frac{\partial \mathcal{L}}{\partial \dot{\vec{q}}}=m\dot{\vec{q}}, \qquad \vec{\varsigma}=\frac{\partial\mathcal{L}}{\partial\dot{\vec{\alpha}}}=\Xi^T\mathbb{I}\Xi\dot{\vec{\alpha}}.
\end{equation}
Notice now that given the fact that $D^2\mathcal{L}$ is positive definite then $\varsigma$ is well-defined even if it depends on both $\gr{\vec{\alpha}}$ and $\gr{\dot{\vec{\alpha}}}$, and we can proceed to define the Hamiltonian of a rigid motion as
\begin{equation}
  \mathcal{H}(\vec{q},\vec{\alpha},\vec{p},\vec{\varsigma}) = \frac{1}{2m}\abs{\vec{p}}^2+\frac{1}{2}\vec{\varsigma}\cdot(\Xi^{T}\mathbb{I}\Xi)^{-1}\vec{\varsigma}.\label{eq:HamiltonianRigidBody}
\end{equation}
Hence our choice of coordinates in \eqref{eq:boltzman} is perfectly admissible.
\section{The Boltzmann--Curtiss equation}
In this section we provide a detailed derivation of the Boltzmann--Curtiss equation for convex non-spherical symmetric-top rigid molecules, using the modern BBGKY hierarchy approach \cite{harris}.
We assume the calamitic rarefied gas under consideration is an ensemble of $N$ convex non-spherical symmetric-top rigid bodies.
The state of each particle can be identified by the coordinates $\vec{\Gamma}=(\vec{q},\vec{\alpha},\vec{p},\vec{\varsigma})$, where $\vec{q}$ represents the position of the center of mass of the particle, $\vec{\alpha}$ are the Euler angles representing its orientation while $\vec{p}$ and $\vec{\varsigma}$ are the conjugate moments to $\vec{q}$ and $\vec{\alpha}$ respectively.
The configuration of the rarefied calamitic fluid at time $t$ is given by $\{\vec{\Gamma}_i^*(t)\}_{i=1}^N=\{(\vec{q}_i(t),\vec{\alpha}_i(t),\vec{p}_i(t),\vec{\varsigma}_i(t))\}_{i=1}^N$. 
We can now introduce the Klimontovich distribution $\pi$ and the reduced particle distribution $f_s$ as follows,
\begin{gather}
  \pi\Big(\{\vec{\Gamma}_i\}_{i=1}^N,t\Big) = \sum_{i=1}^N \delta\Big(\vec{\Gamma}_i-\vec{\Gamma}_i^*(t)\Big),\\
  f_s (\vec{\Gamma}_1,\dots,\vec{\Gamma}_s,t)=\int \pi(\vec{\Gamma}_1,\dots,\vec{\Gamma}_s,\vec{\Gamma}^{(s)},t)\,d\vec{\Gamma}^{(s)},\label{eq:ParticleDist}
\end{gather} 
where $d\vec{\Gamma}^{(s)}=d\vec{\Gamma}_{s+1} \dots d\vec{\Gamma}_{\gr{N}}$.
From \eqref{eq:HamiltonianRigidBody} we know that the Hamiltonian of the particle ensemble is given by
\begin{equation}
  \mathcal{H}= \left(\sum_{i=1}^{N} \frac{\abs{\vec{p}_i}^2}{2m}+\frac{1}{2}\vec{\varsigma}_i\cdot(\Xi^{T}\mathbb{I}\Xi)^{-1}\vec{\varsigma}_i\right)+\sum_{1\leq i < j \leq N}\phi(\vec{q}_i,\vec{q}_j,\vec{\alpha}_i,\vec{\alpha}_j),
\end{equation} 
where $\mathbb{I}$ is the inertia tensor of a single molecule, $\Xi$ is the tensor used to pass from angular velocity to the total time derivative of the Euler angles in \eqref{eq:angularVelocity}, and $\phi:\mathbb{R}^3\to \mathbb{R}$ is the interaction potential we are considering in addition to the Hamiltonian corresponding to the rigid motion described in \eqref{eq:HamiltonianRigidBody}.
We decompose the Hamiltonian into three terms: two containing respectively contributions from particles $1,\dots,s$ and $s+1,\dots, N$ and a third term containing mixed terms,
\begin{gather}
  \mathcal{H}_s= \left(\sum_{i=1}^{s} \frac{\abs{\vec{p}_i}^2}{2m}+\frac{1}{2}\vec{\varsigma}_i\cdot(\Xi^{T}\mathbb{I}\Xi)^{-1}\vec{\varsigma}_i\right)+\sum_{1\leq i < j \leq s}\phi(\vec{q}_i,\vec{q}_j,\vec{\alpha}_i,\vec{\alpha}_j),\\
  \mathcal{H}_{N-s}= \left(\sum_{i=s+1}^{N} \frac{\abs{\vec{p}_i}^2}{2m}+\frac{1}{2}\vec{\varsigma}_i\cdot(\Xi^{T}\mathbb{I}\Xi)^{-1}\vec{\varsigma}_i\right)+\sum_{s+1\leq i < j \leq N}\phi(\vec{q}_i,\vec{q}_j,\vec{\alpha}_i,\vec{\alpha}_j),\\
  \hat{\mathcal{H}}_{s}= \sum_{i=1}^{s} \sum_{j=s+1}^{N} \phi(\vec{q}_i,\vec{q}_j,\vec{\alpha}_i,\vec{\alpha}_j).
\end{gather}
We then differentiate $f_s$ with respect to time and use Liouville's theorem
\begin{equation}
\frac{\partial \pi}{\partial t}=-\{\pi,\mathcal{H}\}
\end{equation}
to obtain the following integro-differential equation for the reduced particle distribution $f_s$,
\begin{equation}
  \frac{\partial f_s}{\partial t} 
  = \int \frac{\partial \pi}{\partial t}\, d\vec{\Gamma}^{(s)}
  =-\int \left(\{\pi,\mathcal{H}_s\}+\{\pi,\mathcal{H}_{n-s}\}+\{\pi,\hat{\mathcal{H}}_s\}\right)\,d\Gamma^{(s)}.
\end{equation}
We notice that since $\mathcal{H}_s$ is independent of $\vec{\Gamma}_{s+1},\dots,\vec{\Gamma}_{n}$ we can bring the integral inside the Poisson brackets to obtain the following equation,
\begin{equation}
  \frac{\partial f_s}{\partial t} 
  = \int \frac{\partial \pi}{\partial t}\, d\vec{\Gamma}^{(s)}
  = - \{f_s,\mathcal{H}_s\}
    - \int \{\pi,\mathcal{H}_{n-s}\} \,d\Gamma^{(s)}
    - \int\{\pi,\hat{\mathcal{H}}_s\}\,d\Gamma^{(s)}.
    \label{eq:Liouville}
\end{equation}
We then notice that the second term on the right-hand side vanishes given the fact it is an exact divergence and we have appropriate boundary conditions:
\begin{align}
  \int\!\!\{\pi&,\mathcal{H}_{N-s}\} d\vec{\Gamma}^{(s)}\!=\!\!\!\int\!\!\sum_{i=1}^{n}\!\Big(\!\frac{\partial \pi}{\partial \vec{q}_i}\frac{\partial \mathcal{H}_{N-s}}{\partial \vec{p}_i}\!+\!\frac{\partial \pi}{\partial \vec{\alpha}_i}\frac{\partial \mathcal{H}_{N-s}}{\partial \vec{\varsigma}_i}\!-\!\frac{\partial \pi}{\partial \vec{p}_i}\frac{\partial \mathcal{H}_{N-s}}{\partial \vec{q}_i}\!-\!\frac{\partial \pi}{\partial \vec{\varsigma}_i}\frac{\partial \mathcal{H}_{N-s}}{\partial \vec{\alpha}_i}\!\Big)d\vec{\Gamma}^{(s)}\\
  &=\int\sum_{i=s+1}^{N}\Big(\frac{\partial \pi}{\partial \vec{q}_i}\frac{\vec{p}_i}{m}+\frac{\partial \pi}{\partial \vec{\alpha}}(\Xi^T\mathbb{I}\Xi)^{-1}\vec{\varsigma}\Big)\,d\vec{\Gamma}^{(s)}\nonumber\\
  &\;\;\;\;-\int\sum_{i=s+1}^{N}\Big(\frac{\partial \pi}{\partial \vec{p}_i}\sum_{j=i+1}^{N}\frac{\partial  \phi(\vec{q}_i,\vec{q}_j,\vec{\alpha}_i,\vec{\alpha}_j)}{\partial\vec{q}_i}+\frac{\partial \pi}{\partial \vec{\varsigma}_i}\sum_{j=i+1}^{N}\frac{\partial \phi(\vec{q}_i,\vec{q}_j,\vec{\alpha}_i,\vec{\alpha}_j)}{\partial\vec{\alpha}_i}\Big)\,d\vec{\Gamma}^{(s)}\nonumber\\
  &=\fnmark\int\sum_{i=s+1}^{N}\Bigg(\frac{\partial }{\partial \vec{q}_i}\cdot  \Big(\pi\frac{\vec{p}_i}{m}\Big)+\frac{\partial}{\partial \vec{\alpha}}\cdot\Big(\pi(\Xi^T\mathbb{I}\Xi)^{-1}\vec{\varsigma}\Big)\Bigg)d\vec{\Gamma}^{(s)}\nonumber\\
  &\;\;\;\;\;\;\;-\int\sum_{i=s+1}^{N}\frac{\partial}{\partial \vec{p}_i}\cdot\Big(\pi\sum_{j=i+1}^{N}\frac{\partial\phi(\vec{q}_i,\vec{q}_j,\vec{\alpha}_i,\vec{\alpha}_j)}{\partial\vec{q}_i}\Big)\,d\vec{\Gamma}^{(s)}\nonumber\\
  &\;\;\;\;\;\;\;-\int\sum_{i=s+1}^{N}\frac{\partial}{\partial \vec{\varsigma}_i}\cdot\Big(\pi\sum_{j=i+1}^{N}\frac{\partial\phi(\vec{q}_i,\vec{q}_j,\vec{\alpha}_i,\vec{\alpha}_j)}{\partial\vec{\alpha}_i}\Big)\,d\vec{\Gamma}^{(s)}=0.\nonumber
\end{align}
\fntext{We are here using the fact that $(\Xi^T\mathbb{I}\Xi)^{-1}\vec{\varsigma}=\dot{\vec{\alpha}}$ together with the fact that $\dot{\vec{\alpha}}$ is independent of $\vec{\alpha}$.}
It remains to study the last term in \eqref{eq:Liouville}:
\begin{align}
  \int\!\{\pi,\hat{\mathcal{H}}_{s}\} d\vec{\Gamma}^{(s)}&=\!\int\!\sum_{i=1}^{N}\!\Big(\frac{\partial \pi}{\partial \vec{q}_i}\frac{\partial \hat{\mathcal{H}}_{s}}{\partial \vec{p}_i}\!+\!\frac{\partial \pi}{\partial \vec{\alpha}_i}\frac{\partial \hat{\mathcal{H}}_{s}}{\partial \vec{\varsigma}_i}\!-\!\frac{\partial \pi}{\partial \vec{p}_i}\frac{\partial \hat{\mathcal{H}}_{s}}{\partial \vec{q}_i}\!-\!\frac{\partial \pi}{\partial \vec{\varsigma}_i}\frac{\partial \hat{\mathcal{H}}_{s}}{\partial \vec{\alpha}_i}\Big)\,d\vec{\Gamma}^{(s)}\\
  &=\int\sum_{i=1}^{s}\frac{\partial \pi}{\partial \vec{p}_i}\cdot\sum_{j=s+1}^{N}\frac{\partial\phi(\vec{q}_i,\vec{q}_j,\vec{\alpha}_i,\vec{\alpha}_j)}{\partial\vec{q}_i}\,d\vec{\Gamma}^{(s)}\nonumber\\
  &+\int\sum_{j=s+1}^{N}\frac{\partial \pi}{\partial \vec{p}_j}\cdot\sum_{i=1}^{s}\frac{\partial\phi(\vec{q}_i,\vec{q}_j,\vec{\alpha}_i,\vec{\alpha}_j)}{\partial\vec{q}_j}\,d\vec{\Gamma}^{(s)}\nonumber\\
  &+\int\sum_{i=1}^{s}\frac{\partial \pi}{\partial \vec{\varsigma}_i}\cdot\sum_{j=s+1}^{N}\frac{\partial\phi(\vec{q}_i,\vec{q}_j,\vec{\alpha}_i,\vec{\alpha}_j)}{\partial\vec{\alpha}_i}\,d\vec{\Gamma}^{(s)}\nonumber\\
  &+\int\sum_{j=s+1}^{N}\frac{\partial \pi}{\partial \vec{\varsigma}_i}\cdot\sum_{i=1}^{s}\frac{\partial\phi(\vec{q}_i,\vec{q}_j,\vec{\alpha}_i,\vec{\alpha}_j)}{\partial\vec{\alpha}_j}\,d\vec{\Gamma}^{(s)}.\nonumber
\end{align}
We notice that the second and last term in the previous equation vanish since they are exact divergences, which leaves us with
\begin{align}
  \int\!\{\pi,\hat{\mathcal{H}}_{s}\} d\vec{\Gamma}^{(s)}\!&= 
  \int\sum_{i=1}^{s}\frac{\partial \pi}{\partial \vec{p}_i}\cdot\sum_{j=s+1}^{N}\frac{\partial\phi(\vec{q}_i,\vec{q}_j,\vec{\alpha}_i,\vec{\alpha}_j)}{\partial\vec{q}_i}\,d\vec{\Gamma}^{(s)}\\
  &+\int\sum_{i=1}^{s}\frac{\partial \pi}{\partial \vec{\varsigma}_i}\cdot\sum_{j=s+1}^{N}\frac{\partial\phi(\vec{q}_i,\vec{q}_j,\vec{\alpha}_i,\vec{\alpha}_j)}{\partial\vec{\alpha}_i}\,d\vec{\Gamma}^{(s)}\nonumber\\
  &=\fnmark (N-s)\int\sum_{i=1}^{s}\frac{\partial \pi}{\partial \vec{p}_i}\cdot\frac{\partial\phi(\vec{q}_i,\vec{q}_{s+1},\vec{\alpha}_i,\vec{\alpha}_{s+1})}{\partial\vec{q}_i}\,d\vec{\Gamma}^{(s)}\nonumber\\
  &+(N-s)\int\sum_{i=1}^{s}\frac{\partial \pi}{\partial \vec{\varsigma}_i}\cdot\frac{\partial\phi(\vec{q}_i,\vec{q}_{s+1},\vec{\alpha}_i,\vec{\alpha}_{s+1})}{\partial\vec{\alpha}_i}\,d\vec{\Gamma}^{(s)}\nonumber\\
  &=\fnmark(N-s)\int\sum_{i=1}^{s}\frac{\partial f_{s+1}}{\partial \vec{p}_i}\cdot\frac{\partial\phi(\vec{q}_i,\vec{q}_{s+1},\vec{\alpha}_i,\vec{\alpha}_{s+1})}{\partial\vec{q}_i}\,d\vec{\Gamma}_{s+1}\nonumber\\
  &+(N-s)\int\sum_{i=1}^{s}\frac{\partial f_{s+1}}{\partial \vec{\varsigma}_i}\cdot\frac{\partial\phi(\vec{q}_i,\vec{q}_{s+1},\vec{\alpha}_i,\vec{\alpha}_{s+1})}{\partial\vec{\alpha}_i}\,d\vec{\Gamma}_{s+1}.\nonumber
\end{align}
\fntext{Here we use the fact that we are summing over $j=s+1,\dots,n$ and this amounts to $n-s$ identical contributions, thanks to the permutation symmetry of $\pi$, which comes from the fact that particles are indistinguishable from one another.}
\fntext{Here we use \eqref{eq:ParticleDist}.}
We have the BBGKY hierarchy for the Boltzmann--Curtiss equation, i.e.~
\begin{align}
  \frac{\partial f_s}{\partial t} + \{f_s,\mathcal{H}_s\} &= (N-s)\int\sum_{i=1}^{s}\frac{\partial f_{s+1}}{\partial \vec{p}_i}\cdot\frac{\partial\phi(\vec{q}_i,\vec{q}_{s+1},\vec{\alpha}_i,\vec{\alpha}_{s+1})}{\partial\vec{q}_i}\,d\vec{\Gamma}_{s+1}\\
  &+(N-s)\int\sum_{i=1}^{s}\frac{\partial f_{s+1}}{\partial \vec{\varsigma}_i}\cdot\frac{\partial\phi(\vec{q}_i,\vec{q}_{s+1},\vec{\alpha}_i,\vec{\alpha}_{s+1})}{\partial\vec{\alpha}_i}\,d\vec{\Gamma}_{s+1}.\nonumber
\end{align}
By abuse of notation, we will denote $f_s$ the function $f_s$ multiplied by the factor $\frac{N!}{(N-s)!}$ to obtain the following expression for the BBGKY hierarchy, i.e.~
\begin{align}
  \frac{\partial f_s}{\partial t} + \{f_s,\mathcal{H}_s\} &=\int\sum_{i=1}^{s}\frac{\partial f_{s+1}}{\partial \vec{p}_i}\cdot\frac{\partial\phi(\vec{q}_i,\vec{q}_{s+1},\vec{\alpha}_i,\vec{\alpha}_{s+1})}{\partial\vec{q}_i}\,d\vec{\Gamma}_{s+1}\\
  &+\int\sum_{i=1}^{s}\frac{\partial f_{s+1}}{\partial \vec{\varsigma}_i}\cdot\frac{\partial\phi(\vec{q}_i,\vec{q}_{s+1},\vec{\alpha}_i,\vec{\alpha}_{s+1})}{\partial\vec{\alpha}_i}\,d\vec{\Gamma}_{s+1}.\nonumber
\end{align}
In particular, we focus our attention on the first two terms of the BBGKY hierarchy, i.e.~
\begin{align}
  \frac{\partial f_1}{\partial t} + \frac{\vec{p}_1}{m}\cdot\frac{\partial f_1}{\partial \vec{q}_1}+(\Xi^T_1 \mathbb{I}\Xi_1)^{-1}\vec{\varsigma}_1\cdot \frac{\partial f_1}{\partial \vec{\alpha}_1} &= \int \frac{\partial \phi(\vec{q}_1,\vec{q}_2,\vec{\alpha}_1,\vec{\alpha}_2)}{\partial \vec{q}_1}\frac{\partial f_2}{\partial\vec{p}_1}\,d\vec{\Gamma}_2\label{eq:BBGKY1}\\
  &+\int\frac{\partial \phi(\vec{q}_1,\vec{q}_2,\vec{\alpha}_1,\vec{\alpha}_2)}{\partial \vec{\alpha}_1}\frac{\partial f_2}{\partial \vec{\varsigma}_1}\,d\vec{\Gamma}_2\nonumber
\end{align}
\begin{align}
  \frac{\partial f_2}{\partial t} + \frac{\vec{p}_1}{m}\cdot\frac{\partial f_2}{\partial \vec{q}_1}&+(\Xi^T_1 \mathbb{I}\Xi_1)^{-1}\vec{\varsigma}_1\cdot \frac{\partial f_2}{\partial \vec{\alpha}_1} + \frac{\vec{p}_2}{m}\cdot \frac{\partial f_2}{\partial \vec{q}_2}+(\Xi^T_2 \mathbb{I}\Xi_2)^{-1}\vec{\varsigma}_2\cdot\frac{\partial f_2}{\partial \vec{\alpha}_2}\label{eq:BBGKY2}\\
  &- \frac{\partial\phi(\vec{q}_1,\vec{q}_2,\vec{\alpha}_1,\vec{\alpha}_2)}{\partial \vec{q}_1}\frac{\partial f_2}{\partial\vec{p}_1} - \frac{\partial \phi(\vec{q}_1,\vec{q}_2,\vec{\alpha}_1,\vec{\alpha}_2)}{\partial \vec{\alpha}_1}\frac{\partial f_2}{\partial \vec{\varsigma}_1}=0\nonumber
\end{align}
where the right-hand side of the last equation vanishes because we have assumed $f_3\equiv 0$ since we are working under the assumption that we only have binary collisions.
Indeed under the Boltzmann--Grad limit it is possible to justify more precisely that $f_3$ is negligible; we redirect the interested reader to \cite{huangKerson,kardar, harris}.
Notice now that we can add under the integral appearing on the right-hand side of \eqref{eq:BBGKY1} and \eqref{eq:BBGKY2} terms of the form $\frac{\partial \phi(\vec{q}_1,\vec{q}_2,\vec{\alpha}_1,\vec{\alpha}_2)}{\partial \vec{q}_2}\frac{\partial f_2}{\partial\vec{p}_2}$ and $\frac{\partial \phi(\vec{q}_1,\vec{q}_2,\vec{\alpha}_1,\vec{\alpha}_2)}{\partial \vec{\alpha}_2}\frac{\partial f_2}{\partial \vec{\varsigma}_2}$ respectively, since they will vanish due to the fact that they are exact divergences.
We focus our attention on the first two terms of the BBGKY hierarchy, i.e.~
\begin{align}
  \frac{\partial f_1}{\partial t} + \frac{\vec{p}_1}{m}\cdot\frac{\partial f_1}{\partial \vec{q}_1}+(\Xi^T_1 \mathbb{I}\Xi_1)^{-1}\vec{\varsigma}_1\frac{\partial f_1}{\partial \vec{\alpha}_1} &= \int \frac{\partial \phi(\vec{q}_1,\vec{q}_2,\vec{\alpha}_1,\vec{\alpha}_2)}{\partial \vec{q}_1}\Big(\frac{\partial f_2}{\partial\vec{p}_1}-\frac{\partial f_2}{\partial \vec{p}_2}\Big)\,d\vec{\Gamma}_2\label{eq:BBGKY12}\\
  &+\int\frac{\partial \phi(\vec{q}_1,\vec{q}_2,\vec{\alpha}_1,\vec{\alpha}_2)}{\partial \vec{\alpha}_1}\frac{\partial f_2}{\partial \vec{\varsigma}_1}d\vec{\Gamma}_2\nonumber\\
  &+\int\frac{\partial \phi(\vec{q}_1,\vec{q}_2,\vec{\alpha}_1,\vec{\alpha}_2)}{\partial \vec{\alpha}_2}\frac{\partial f_2}{\partial \vec{\varsigma}_2}d\vec{\Gamma}_2,\nonumber\\
  \frac{\partial f_2}{\partial t} + \frac{\vec{p}_1}{m}\cdot \frac{\partial f_2}{\partial \vec{q}_1}+(\Xi^T_1 \mathbb{I}\Xi_1)^{-1}\vec{\varsigma}_1\cdot\frac{\partial f_2}{\partial \vec{\alpha}_1} &+ \frac{\vec{p}_2}{m}\cdot\frac{\partial f_2}{\partial \vec{q}_2}+(\Xi^T_2 \mathbb{I}\Xi_2)^{-1}\vec{\varsigma}_2\cdot\frac{\partial f_2}{\partial \vec{\alpha}_2}\label{eq:BBGKY22}\\
  &- \frac{\partial\phi(\vec{q}_1,\vec{q}_2,\vec{\alpha}_1,\vec{\alpha}_2)}{\partial \vec{q}_1}\Big(\frac{\partial f_2}{\partial\vec{p}_1}-\frac{\partial f_2}{\partial \vec{p}_2}\Big) \nonumber\\
  &- \frac{\partial \phi(\vec{q}_1,\vec{q}_2,\vec{\alpha}_1,\vec{\alpha}_2)}{\partial \vec{\alpha}_1}\frac{\partial f_2}{\partial \vec{\varsigma}_1}\nonumber\\
  &- \frac{\partial \phi(\vec{q}_1,\vec{q}_2,\vec{\alpha}_1,\vec{\alpha}_2)}{\partial \vec{\alpha}_2}\frac{\partial f_2}{\partial \vec{\varsigma}_2}=0,\nonumber
\end{align}
We notice that there are two different length scales in both \eqref{eq:BBGKY1} and \eqref{eq:BBGKY2}.
In particular, the left-hand side of \eqref{eq:BBGKY1} acts on the macroscopic length scale while the right-hand side acts on the microscopic length scale.
We introduce fast varying coordinates, i.e.~$\vec{q} = \vec{q}_2-\vec{q}_1$ and $\vec{\sigma}=\vec{\alpha}_2-\vec{\alpha}_1$ together with slowly varying coordinates, i.e.~$\vec{Q}=\frac{\vec{q}_2+\vec{q}_1}{2}$ and $\vec{\Sigma}=\frac{\vec{\alpha}_2+\vec{\alpha}_1}{2}$.
Now we can rewrite \eqref{eq:BBGKY2} in terms of the new coordinates as,
\begin{align}
  \frac{\partial f_2}{\partial t}\!+ \frac{1}{2}\frac{\vec{p}_2\!+\!\vec{p}_1}{m}\cdot\frac{\partial f_2}{\partial \vec{Q}}+\gr{\vec{A}}\cdot\frac{\partial f_2}{\partial {\vec{\Sigma}}}\boxed{\!+\frac{\vec{p}_2\!-\!\vec{p}_1}{m}\cdot\frac{\partial f_2}{\partial \vec{q}}\!+\!\gr{\vec{B}}\cdot\frac{\partial f_2}{\partial \vec{\sigma}}}&\label{eq:BBGKY2Separated}\\
  \boxed{-\frac{\partial \phi(\vec{q}_1,\vec{q}_2,\vec{\alpha}_1,\vec{\alpha}_2)}{\partial \vec{q}_1}\Big(\frac{\partial f_2}{\partial\vec{p}_1}-\frac{\partial f_2}{\partial \vec{p}_2}\Big)}&\nonumber\\
  \boxed{-\frac{\partial \phi(\vec{q}_1,\vec{q}_2,\vec{\alpha}_1,\vec{\alpha}_2)}{\partial \vec{\alpha}_1}\frac{\partial f_2}{\partial \vec{\varsigma}_1}-\frac{\partial \phi(\vec{q}_1,\vec{q}_2,\vec{\alpha}_1,\vec{\alpha}_2)}{\partial \vec{\alpha}_2}\frac{\partial f_2}{\partial \vec{\varsigma}_2}}\!&=\!0,\nonumber
\end{align}
where ${\vec{A}=\frac{1}{2}\left((\Xi_1^T\mathbb{I}\Xi_1)^{-1}\vec{\varsigma}_1+(\Xi_2^T\mathbb{I}\Xi_2)^{-1}\vec{\varsigma}_2\right)}$, ${\vec{B}=\left((\Xi_2^T\mathbb{I}\Xi_2)^{-1}\vec{\varsigma}_2-(\Xi_1^T\mathbb{I}\Xi_1)^{-1}\vec{\varsigma}_1\right)}$ and we have boxed all the terms acting on the microscopic scale length. We proceed to treat the terms acting on different scale lengths separately, and we notice that the terms in the previous equation corresponding to the microscopic scale length can be treated as steady with respect to the collisional time scale,
\begin{align}
  \frac{\vec{p}_2\!-\!\vec{p}_1}{m}\cdot \frac{\partial f_2}{\partial \vec{q}}+\gr{\vec{B}}\cdot\frac{\partial f_2}{\partial \vec{\sigma}}\label{eq:BBGKYSimply}
  =\frac{\partial\phi(\vec{q}_1,\vec{q}_2,\vec{\alpha}_1,\vec{\alpha}_2)}{\partial \vec{q}_1}\Big(\frac{\partial f_2}{\partial\vec{p}_1}-\frac{\partial f_2}{\partial \vec{p}_2}\Big)&\\
  +\frac{\partial \phi(\vec{q}_1,\vec{q}_2,\vec{\alpha}_1,\vec{\alpha}_2)}{\partial \vec{\alpha}_1}\frac{\partial f_2}{\partial \vec{\varsigma}_1}+\frac{\partial \phi(\vec{q}_1,\vec{q}_2,\vec{\alpha}_1,\vec{\alpha}_2)}{\partial \vec{\alpha}_2}\frac{\partial f_2}{\partial \vec{\varsigma}_2}.&\nonumber
\end{align}
Substituting \eqref{eq:BBGKYSimply} into \eqref{eq:BBGKY1} we obtain,
\begin{align}
  &\frac{\partial f_1}{\partial t} + \frac{\vec{p}_1}{m}\cdot\frac{\partial f_1}{\partial \vec{q}_1}+(\Xi^T_1 \mathbb{I}\Xi_1)^{-1}\vec{\varsigma}_1\frac{\partial f_1}{\partial \vec{\alpha}_1} = \left. \frac{d \,f_1}{d\,t}\right\lvert_{\vspace{0.5pt} \text{collision}},\label{eq:PreBoltzmann}\\
  &\left. \frac{d \,f_1}{d\,t}\right\lvert_{\vspace{0.5pt} \text{collision}} \!\!\!\!= \int \left(\frac{\vec{p}_2\!-\!\vec{p}_1}{m}\cdot\frac{\partial f_2}{\partial \vec{q}}+\gr{\vec{B}}\cdot\frac{\partial f_2}{\partial \vec{\sigma}}\,d\vec{\Gamma}_2\right),\label{eq:PreCollision}
\end{align}
where the right-hand side of \eqref{eq:PreCollision} represents the total variation of $f_1$ due to collisions between two molecules.
\section{The rigid collision operator}
\label{sec:BoltzmannCurtissCollision}
In the previous section, we derived a hierarchy of equations to describe the evolution of the reduced distribution functions $f_s$.
In particular, each equation in the hierarchy depended on a higher-order reduced particle distribution, hence the system of equations was not closed. To solve this issue we removed the dependence on $f_3$ from the second equation in the hierarchy, i.e.~\eqref{eq:BBGKY2}, by using the physical assumption that we are only considering binary collisions.
We will now focus our attention on a way to remove the dependence on $f_2$ from the first equation in the hierarchy, i.e.~\eqref{eq:BBGKY1} when considering the collision between two hard convex molecules.

Notice that since $f_3$ vanishes we are only considering collisions between two molecules. We will from now on label these two molecules as molecule 1 and molecule 2 and use the notation $\vec{\Gamma}_1$ and $\vec{\Gamma}_2$ to denote the phase space coordinates of the two molecules respectively.
In particular, we will consider a reference frame having origin in the center of mass of molecule 1 and imagine that molecule 2 is approaching, colliding and then moving away from molecule 1.

We begin expanding $d\vec{\Gamma}_2$ in \eqref{eq:PreCollision} and making use of the fact that since we are only considering collisions between rigid bodies $f_2$ will vanish outside of the translated excluded volume $V=V_{ex}^{1,2}(\sigma)+\vec{q}_1$, hence we can restrict the domain of integration to $V$, i.e.~
\begin{align}
  \left. \frac{d \,f_1}{d\,t}\right\lvert_{\vspace{0.5pt} \text{collision}}\!\!\!\!\!\!\!&= \int d\vec{\varsigma}_2d\vec{p}_2d\vec{\alpha}_2\int_{V} \frac{\vec{p}_2\!-\!\vec{p}_1}{m}\cdot\frac{\partial f_2}{\partial \vec{q}}\,\gr{d\vec{q}}+\!\!\int d\vec{\varsigma}_2d\vec{p}_2d\vec{\alpha}_2\!\!\int_{V}\!\!\!\gr{\vec{B}}\cdot\frac{\partial f_2}{\partial \vec{\sigma}}\,\gr{d\vec{q}}\\
  &=\!\!\! \int d\vec{\varsigma}_2d\vec{p}_2d\vec{\alpha}_2\int_{\partial V} \frac{\vec{p}_2\!-\!\vec{p}_1}{m}\cdot\vec{k}f_2d\vec{S}+\!\!\int d\vec{\varsigma}_2d\vec{p}_2d\vec{\alpha}_2\int_{V}\!\!\!\gr{\vec{B}}\cdot\frac{\partial f_2}{\partial \vec{\sigma}}\,d\gr{\vec{q}}.\;\;\label{eq:rigidCollision}
\end{align}
We now observe that the translated excluded volume $V$ is a strictly convex set with smooth boundary, hence we can apply results from differential geometry of surfaces with positive curvature.
In particular, we can transform the integral on the boundary of $V$ into an integral on the surface of the unit sphere $\mathbb{S}^2$, making use of the element of surface area per solid angle $dS$ that can be computed explicitly \cite{curtissV}, and rewrite \eqref{eq:rigidCollision} as
\begin{align}
  \left. \frac{d \,f_1}{d\,t}\right\lvert_{\vspace{0.5pt} \text{collision}} \!\!\!\!&= \int d\vec{\varsigma}_2d\vec{p}_2d\vec{\alpha}_2 \int_{\mathbb{S}^2}  \frac{\vec{p}_2\!-\!\vec{p}_1}{m}\cdot\vec{k}f_2dS\,d\vec{k}\\
  &+\int d\vec{\varsigma}_2d\vec{p}_2d\vec{\alpha}_2\int_{V}\gr{\vec{B}}\cdot\frac{\partial f_2}{\partial \vec{\sigma}}\,\gr{d\vec{q}}.
\end{align}
We now need to focus on the second term in the previous equation. In particular we need to rewrite it in terms of the surface element of the unit sphere $\mathbb{S}^2$. Notice that $\gr{\vec{B}= \dot{\vec{\sigma}}}$, and so we can rewrite the second term in the previous equation as
\begin{align}
  \int \!\!d\vec{\varsigma}_2d\vec{p}_2d\vec{\alpha}_2\int_{V}\dot{\vec{\sigma}}\cdot\frac{\partial f_2}{\partial \vec{\sigma}}\,d\vec{q}_2 &=\!\! \int d\vec{\varsigma}_2d\vec{p}_2d\vec{\alpha}_2\int_{\mathbb{S}^2}\dot{\vec{\sigma}}\gr{\frac{\partial\vec{q}}{\partial\vec{\sigma}}}\cdot\frac{\partial f_2}{\partial \vec{q}}\,\gr{d\vec{q}}\\
\end{align}
 By virtue of {proposition \ref{thm:totalDerivativeDirector}} we can express $\gr{\dot{\vec{\sigma}}\frac{\partial \vec{q}}{\vec{\partial \sigma}}}$ as $\vec{\omega}_1\times\vec{g}_1-\vec{\omega}_2\times\vec{g}_2$, where $\vec{g}_1$ and $\vec{g}_2$ are the vectors connecting $\vec{q}_1$ and $\vec{q}_2$ with the contact point $\vec{\zeta}$. In particular we can rewrite \eqref{eq:rigidCollision} as
 \begin{equation}
  \left. \frac{d \,f_1}{d\,t}\right\lvert_{\vspace{0.5pt} \text{collision}} \!\!\!\! = \int d\vec{\varsigma}_2d\vec{p}_2d\vec{\alpha}_2\int_{\mathbb{S}^2}  \vec{\mathfrak{g}}\cdot\vec{k}f_2\,dS\,d\vec{k},\label{eq:preAnsatzCollision}
 \end{equation}
 where $\vec{\mathfrak{g}}$ is the relative velocity of the contact point $\vec{\zeta}$, i.e.~
 \begin{equation}
  \color{orange}\vec{\mathfrak{g}} = \frac{1}{m}\left(\vec{p}_1-\vec{p}_2\right)+\vec{\omega}_1\times\vec{g}_1-\vec{\omega}_2\times\vec{g}_2.
 \end{equation}
We decompose $\mathbb{S}^2$ into two surfaces: $\mathbb{S}_+$ is the portion of the surface $\mathbb{S}^2$ such that along the collision trajectory molecule 2 is coming towards molecule 1, i.e.~$\mathfrak{g}\cdot \vec{k}$ is positive, while $\mathbb{S}_-$ is the portion of the surface $\mathbb{S}^2$ such that along the collision trajectory molecule 2 is moving away from molecule 1, i.e.~$\mathfrak{g}\cdot \vec{k}$ is negative.
With this we split the integral in \eqref{eq:preAnsatzCollision} into two integrals on $\mathbb{S}_+$ and $\mathbb{S}_-$ respectively, i.e.~
\begin{align}
  \left. \frac{d \,f_1}{d\,t}\right\lvert_{\vspace{0.5pt} \text{collision}} \!\!\!\! &= \int d\vec{\varsigma}_2d\vec{p}_2d\vec{\alpha}_2\int_{\mathbb{S}_+}  \vec{\mathfrak{g}}\cdot\vec{k}f_2\,dS\,d\vec{k}+\int d\vec{\varsigma}_2d\vec{p}_2d\vec{\alpha}_2\int_{\mathbb{S}_-}  \vec{\mathfrak{g}}\cdot\vec{k}f_2\,dS\,d\vec{k}.\label{eq:gainLossCollision}
\end{align}
The first term in the previous equation is known as the \emph{gain term}, while the second term is known as the \emph{loss term}.
Notice now that we can make use of the \emph{molecular chaos assumption} \cite{maxwell}, i.e.~$f_2(\vec{\Gamma}_1,\vec{\Gamma}_2,t)=f_1(\vec{\Gamma}_1,t)f_1(\vec{\Gamma}_2,t)$ on the loss term to obtain
\begin{align}
  \left.\frac{df_1}{dt}\right\vert_{\vspace{0.5pt}\text{collision}} \!\!\!\!&= \int\! d\vec{p}_2\,d\vec{\varsigma}_2\,d\vec{\alpha}\!\int_{\mathbb{S}_{+}}\!\!(\vec{k}\cdot \vec{\mathfrak{g}})f_2\,dSd\vec{k} \\
  &+ \int\!d\vec{p}_2\,d\vec{\varsigma}_2\,d\vec{\alpha}\!\int_{\mathbb{S}_{-}}\!\!(\vec{k}\cdot\vec{g})f_1(\vec{q}_1,\vec{\alpha}_1,\vec{p}_1,\vec{\varsigma}_1)f_1(\vec{q}_2,\vec{\alpha}_2,\vec{p}_2,\vec{\varsigma}_2)\,S(\vec{k})d\vec{k}.\nonumber
\end{align}
As discussed in \cite{stronge}, since energy and linear and angular momentum are conserved after the collision we can trace back the post-collisional coordinates in the gain term of \eqref{eq:gainLossCollision} to the corresponding pre-collisional coordinates. Applying the molecular chaos assumption once again we obtain,
\begin{align}
  \left.\frac{df_1}{dt}\right\vert_{\vspace{0.5pt}\text{collision}}\!\!\!\! &= \int\! d\vec{p}_2\,d\vec{\varsigma}_2\,d\vec{\alpha}\!\int_{\mathbb{S}_{+}}\!\!(\vec{k}\cdot \vec{\mathfrak{g}})f_1(\vec{q}_1^{\prime},\vec{\alpha}_1^{\prime},\vec{p}_1^{\prime},\vec{\varsigma}_1^{\prime},t)f_1(\vec{q}_2^{\prime},\vec{\alpha}_2^{\prime},\vec{p}_2^{\prime},\vec{\varsigma}_2^{\prime},t)\,dSd\vec{k}\\
  &+ \int\!d\vec{p}_2\,d\vec{\varsigma}_2\,d\vec{\alpha}\!\int_{\mathbb{S}_{-}}\!\!(\vec{k}\cdot\vec{\mathfrak{g}})f_1(\vec{q}_1,\vec{\alpha}_1,\vec{p}_1,\vec{\varsigma}_1,t)f_1(\vec{q}_2,\vec{\alpha}_2,\vec{p}_2,\vec{\varsigma}_2,t)\,dSd\vec{k}.\nonumber
\end{align}
Lastly, we shorten the notation by writing the collision operator in the familiar way:
\begin{gather}
  \left.\frac{df_1}{dt}\right\vert_{\vspace{0.5pt} \text{collision}} = C[f_1,f_1],
\end{gather}
with
\begin{equation}
C[f,g] \coloneqq  \int\!\!\!\!\int\!\!\!\!\int\!\!\!\!\int (g^\prime f^\prime-gf)(\vec{k}\cdot\vec{\mathfrak{g}})dSd\vec{k}d\vec{p}_2d\vec{\alpha}_2d\vec{\varsigma}_2,
\end{equation}
where the arguments of $f$ are $(\vec{q}_1,\vec{\alpha}_1,\vec{p}_1,\vec{\varsigma}_1,t)$ and the arguments of $g$ are $(\vec{q}_2,\vec{\alpha}_2,\vec{p}_2,\allowbreak \vec{\varsigma}_2,t)$.
\section*{Acknowledgments} The authors would like to acknowledge E.~Virga, whose insights made this work possible.
{
The authors would also like to thank J.~M\'alek and O.~Sou\v{c}ek for their help in correcting some mistakes in an earlier draft of this work.
}
PEF acknowledges support from the UK Engineering and Physical Sciences Research Council [EPSRC grants EP/R029423/1 and EP/W026163/1].
\bibliographystyle{siamplain}
\bibliography{references}

\begin{thebibliography}{10}

\bibitem{allenEtAll}
{\sc M.~Allen, G.~Evans, D.~Frenkel, and B.~Mulder}, {\em Hard convex body fluids}, Advances in Chemical Physics, 86 (1993), p.~1.

\bibitem{ball}
{\sc J.~M. Ball}, {\em Mathematics and liquid crystals}, Molecular Crystals and Liquid Crystals, 647 (2017), pp.~1--27.

\bibitem{bertolottiEtAll}
{\sc M.~Bertolotti, S.~Martellucci, F.~Scudieri, and D.~Sette}, {\em Acoustic modulation of light by nematic liquid crystals}, Applied Physics Letters, 21 (1972), pp.~74--75.

\bibitem{caiEtAll}
{\sc Y.~Cai, P.~Zhang, and A.~C. Shi}, {\em Liquid crystalline bilayers self-assembled from rod–coil diblock copolymers}, Soft Matter, 13 (2017), pp.~4607--4615.

\bibitem{chase}
{\sc M.~W. Chase and N.~I. S.~O. (US)}, {\em NIST-JANAF thermochemical tables}, vol.~9, American Chemical Society Washington, DC, 1998.

\bibitem{nollColeman}
{\sc B.~D. Coleman and W.~Noll}, {\em The thermodynamics of elastic materials with heat conduction and viscosity}, Archive for Rational Mechanics and Analysis, 13 (1963), pp.~167--178.

\bibitem{curtissI}
{\sc C.~F. Curtiss}, {\em Kinetic theory of nonspherical molecules}, The Journal of Chemical Physics, 24 (1956), pp.~225--241.

\bibitem{curtissV}
{\sc C.~F. Curtiss and J.~S. Dahler}, {\em Kinetic theory of nonspherical molecules. {V}}, The Journal of Chemical Physics, 38 (1963), pp.~2352--2363.

\bibitem{curtissII}
{\sc C.~F. Curtiss and C.~Muckenfuss}, {\em Kinetic theory of nonspherical molecules. {II}}, The Journal of Chemical Physics, 26 (1957), pp.~1619--1636.

\bibitem{dahlerSatherI}
{\sc J.~S. Dahler and N.~F. Sather}, {\em Kinetic theory of loaded spheres. {I}}, The Journal of Chemical Physics, 38 (1963), pp.~2363--2382.

\bibitem{deGennes}
{\sc P.~G. {de Gennes} and J.~Prost}, {\em The Physics of Liquid Crystals}, Oxford Science Publications, Clarendon Press, 2nd~ed., 1993.

\bibitem{vanKampen}
{\sc P.~J. Dellar}, {\em {Macroscopic descriptions of rarefied gases from the elimination of fast variables}}, Physics of Fluids, 19 (2007), p.~107101.

\bibitem{ericksen2}
{\sc J.~L. Ericksen}, {\em Transversely isotropic fluids}, Colloid and Polymer Science, 173 (1960), pp.~117--122.

\bibitem{fasanoMarmi}
{\sc A.~Fasano, S.~Marmi, and B.~Pelloni}, {\em Analytical Mechanics: an Introduction}, Oxford Graduate Texts, Oxford University Press, 2013.

\bibitem{frank}
{\sc F.~C. Frank}, {\em {I. Liquid crystals. On the theory of liquid crystals}}, Discussions of the Faraday Society, 25 (1958), p.~19.

\bibitem{gelbartBenShaul}
{\sc W.~M. Gelbart and A.~Ben-Shaul}, {\em Molecular theory of curvature elasticity in nematic liquids}, The Journal of Chemical Physics, 77 (1982), pp.~916--933.

\bibitem{goldsteinPooleSafko}
{\sc H.~Goldstein, C.~P. Poole, and J.~L. Safko}, {\em Classical Mechanics}, Pearson Education, 3rd~ed., 2014.

\bibitem{gonzalezStuart}
{\sc O.~Gonzalez and A.~M. Stuart}, {\em A First Course in Continuum Mechanics}, Cambridge Texts in Applied Mathematics, Cambridge University Press, 2008.

\bibitem{harris}
{\sc S.~Harris}, {\em An Introduction to the Theory of the Boltzmann Equation}, Courier Corporation, 2004.

\bibitem{huangKerson}
{\sc K.~Huang}, {\em Statistical Mechanics}, Wiley, 2nd~ed., 1987.

\bibitem{kardar}
{\sc M.~Kardar}, {\em Statistical Physics of Particles}, Cambridge University Press, 2007.

\bibitem{kitstonGeisow}
{\sc S.~Kitson and A.~Geisow}, {\em Controllable alignment of nematic liquid crystals around microscopic posts: Stabilization of multiple states}, Applied Physics Letters, 80 (2002), pp.~3635--3637.

\bibitem{leslie2}
{\sc F.~M. Leslie}, {\em Continuum theory for nematic liquid crystals}, Continuum Mechanics and Thermodynamics, 4 (1992), pp.~167--175.

\bibitem{curtissIV}
{\sc P.~M. Livingston and C.~F. Curtiss}, {\em Kinetic theory of nonspherical molecules. {IV. A}ngular momentum transport coefficient}, The Journal of Chemical Physics, 31 (1959), pp.~1643--1645.

\bibitem{majumdarEtAll}
{\sc A.~Majumdar, C.~J.~P. Newton, J.~M. Robbins, and M.~Zyskin}, {\em Topology and bistability in liquid crystal devices}, Physics Review E, 75 (2007), p.~051703.

\bibitem{maxwell}
{\sc J.~C. Maxwell}, {\em {IV. O}n the dynamical theory of gases}, Philosophical transactions of the Royal Society of London, 157 (1867), pp.~49--88.

\bibitem{McCourt}
{\sc F.~R.~W. McCourt}, {\em Nonequilibrium Phenomena in Polyatomic Gases}, International Series of Monographs on Chemistry, Clarendon Press, 1990--1991.
\newblock Two volumes.

\bibitem{curtissIII}
{\sc C.~Muckenfuss and C.~F. Curtiss}, {\em Kinetic theory of nonspherical molecules. {III}}, The Journal of Chemical Physics, 29 (1958), pp.~1257--1272.

\bibitem{heatCapacity}
{\sc P.~K. Mukherjee}, {\em Anomalous heat capacity of nematic liquid crystals}, The Journal of Chemical Physics, 109 (1998), pp.~3701--3702.

\bibitem{oseen}
{\sc C.~W. Oseen}, {\em The theory of liquid crystals}, Transactions of the Faraday Society, 29 (1933), p.~883.

\bibitem{dahlerSandlerII}
{\sc S.~I. Sandler and J.~S. Dahler}, {\em Kinetic theory of loaded spheres. {II}}, The Journal of Chemical Physics, 43 (2004), pp.~1750--1759.

\bibitem{scudieriEtAll}
{\sc F.~Scudieri, A.~Ferrari, M.~Bertolotti, and D.~Apostol}, {\em Opto-acoustic modulator with a nematic liquid crystal}, Optics Communications, 15 (1975), pp.~57--59.

\bibitem{sonnetVirga}
{\sc A.~M. Sonnet and E.~G. Virga}, {\em Dissipative Ordered Fluids: Theories for Liquid Crystals}, Springer, 2012.

\bibitem{stewart}
{\sc I.~W. Stewart}, {\em The Static and Dynamic Continuum Theory of Liquid Crystals: A Mathematical Introduction}, Liquid Crystals Book Series, 2004.

\bibitem{straley}
{\sc J.~P. Straley}, {\em Frank elastic constants of the hard-rod liquid crystal}, Physical Review A, 8 (1973), pp.~2181--2183.

\bibitem{stronge}
{\sc W.~J. Stronge}, {\em Impact Mechanics}, Cambridge University Press, 2nd~ed., 2018.

\bibitem{sylvester}
{\sc J.~J. Sylvester}, {\em {XIX. A} demonstration of the theorem that every homogeneous quadratic polynomial is reducible by real orthogonal substitutions to the form of a sum of positive and negative squares}, The London, Edinburgh, and Dublin Philosophical Magazine and Journal of Science, 4 (1852), pp.~138--142.

\bibitem{truesdell2}
{\sc C.~Truesdell and R.~G. Muncaster}, {\em Fundamentals of {M}axwell's Kinetic Theory of a Simple Monatomic Gas}, vol.~83 of Pure and Applied Mathematics, Academic Press, Inc. New York-London, 1980.

\bibitem{virga}
{\sc E.~G. Virga}, {\em Variational Theories for Liquid Crystals}, vol.~8 of Applied Mathematics and Mathematical Computation, Chapman and Hall/CRC, New York, 1994.

\bibitem{whittaker}
{\sc E.~T. Whittaker}, {\em A Treatise on The Analytical Dynamics of Particles and Rigid Bodies: with an Introduction to the Problem of Three Bodies}, Cambridge Mathematical Library, Cambridge University Press, 4th~ed., 1988.

\end{thebibliography}
\end{document}